\documentclass[11pt]{article}
\usepackage[left=0.75in,right=0.75in,top=0.75in,bottom=0.75in]{geometry}
\usepackage{comment}
\usepackage{amsmath}
\usepackage{amsthm}
\usepackage{amssymb}
\usepackage[pdftex,
colorlinks,
linkcolor=blue,%
pdftitle=,%
pdfauthor=HSTX,%
pdfkeywords=Multi-Collison]{hyperref}
\usepackage[table,dvipsnames]{xcolor}
\usepackage{tikz}
\usetikzlibrary{decorations.pathmorphing}
\usepackage{fixmath}
\usepackage{booktabs}
\usepackage{braket}
\usepackage{graphicx}
\usepackage{here}
\usepackage{algorithm, algorithmic}

\newcommand*{\etal}{\textit{et~al.}}

\newcommand*{\algo}[1]{\ensuremath{\mathsf{#1}}}

\newcommand*{\event}[1]{\ensuremath{\mathsf{#1}}}

\newcommand*{\card}[1]{\lvert #1 \rvert}

\newcommand{\relmiddle}[1]{\mathrel{}\middle#1\mathrel{}}
\newcommand{\mymiddle}{\relmiddle{|}}

\newcommand*{\Time}{T}
\newcommand*{\Memory}{S}

\usepackage{xspace}

\newcommand*{\Mclaw}{\algo{Mclaw}}

\newcommand*{\Func}{\algo{Func}}
\newcommand*{\BBHT}{\algo{BBHT}}
\newcommand*{\E}{{\mathrm{\bf E}}}

\newcommand*{\Img}{\mathrm{Im}}

\newcommand*{\A}{\mathcal{A}}
\newcommand*{\B}{\mathcal{B}}

\newcommand*{\BHT}{\algo{BHT}}

\newcommand*{\MTPS}{\algo{MTQS}}

\newcommand*{\good}{\event{good}}

\newcommand*{\pregood}{\event{pregood}}
\newcommand*{\equal}{\event{equal}}
\newcommand{\upp}[2]{{#1}^{(#2)}}
\renewcommand{\set}[1]{\{#1\}}
\newcommand*{\rcoll}{\mathsf{RecMColl}}
\newcommand{\ceil}[1]{\lceil #1 \rceil}
\newcommand{\Natural}{\mathbb{N}}

\newcommand{\reg}[1]{\mathsf{#1}}

\usepackage{color}
\newcommand{\hnote}[1]{#1}

\newcommand{\tnote}[1]{#1}

\renewcommand{\algorithmicensure}{\textbf{Output:}}

\def\equationautorefname~#1\null{eq.~(#1)\null}

\newcommand{\email}[1]{\tt #1}

\newtheorem{theorem}{Theorem}

\newtheorem{problem}{Problem}

\newtheorem{lemma}{Lemma}

\newtheorem{corollary}{Corollary}

\newtheorem{claim}{Claim}

\newenvironment{proofof}[1]{\noindent{\it Proof of }#1.}{\hfill$\square$}

\newcommand{\calO}{{\cal O}}

\begin{document}
\begin{titlepage}
\pagestyle{plain}

\title{\bf Quantum Algorithm for the Multicollision Problem%
\footnote{Preliminary versions of this paper appeared in the proceedings
of Asiacrypt 2017~\cite{DBLP:conf/asiacrypt/HosoyamadaSX17} and PQCrypto 2019~\cite{PQC19}.}
}

\author{
Akinori Hosoyamada\footnotemark[2]$^{~,}$\footnotemark[3] \and
Yu Sasaki\footnotemark[2] \and
Seiichiro Tani\footnotemark[4] \and
Keita Xagawa\footnotemark[2]
}
\date{\today}
\maketitle

\renewcommand{\thefootnote}{\fnsymbol{footnote}}
\begin{center}
\footnotemark[2] {NTT Secure Platform Laboratories, NTT Corporation. \\
 3-9-11, Midori-cho, Musashino-shi, Tokyo 180-8585, Japan.}
 \and
 
 \footnotemark[3]
 {Department of Information and Communication Engineering, Nagoya University.\\
 Furo-cho, Chikusa-ku, Nagoya 464-8603, Japan.}
 \and
 
 \footnotemark[4]
 {NTT Communication Science Laboratories, NTT Corporation. \\
 3-1, Morinosato-Wakamiya, Atsugi-shi, Kanagawa 243-0198, Japan.}
 \email{\{akinori.hosoyamada.bh,yu.sasaki.sk,seiichiro.tani.cs,keita.xagawa.zv\}@hco.ntt.co.jp}
\end{center}

\renewcommand{\thefootnote}{\arabic{footnote}}

\begin{abstract}
The current paper presents a new quantum algorithm for finding multicollisions, often denoted by $\ell$-collisions, where an $\ell$-collision for a function is a set of $\ell$ distinct inputs that are mapped
by the function to the same value.
In cryptology, it is important to study how many queries are required to find an $\ell$-collision for a random function of which domain is larger than its range.
However, the problem of finding $\ell$-collisions for random functions has not received much attention
in the quantum setting.
The tight bound of quantum query complexity for finding a $2$-collisions of a random function has been revealed to be $\Theta(N^{1/3})$, where $N$ is the size of the range of the function, but neither the lower nor upper bounds are known for general $\ell$-collisions.
The paper first integrates the results from existing research to derive several new observations, e.g.,~$\ell$-collisions can be generated only with $O(N^{1/2})$ quantum queries for any integer constant $\ell$.
It then provides a quantum algorithm that finds an $\ell$-collision for a random function with the average quantum query complexity of $O(N^{(2^{\ell-1}-1) / (2^{\ell}-1)})$, which matches the tight bound of $\Theta(N^{1/3})$ for $\ell=2$ and improves upon the known bounds, including the above simple bound of $O(N^{1/2})$. 
More generally, the  algorithm achieves the average quantum query complexity of $O\big(c_N \cdot N^{({2^{\ell-1}-1})/({ 2^{\ell}-1})}\big)$
and runs 
over $\tilde{O}\big(c_N \cdot N^{({2^{\ell-1}-1})/({ 2^{\ell}-1})}\big)$ qubits
in $\tilde{O}\big(c_N \cdot N^{({2^{\ell-1}-1})/({ 2^{\ell}-1})}\big)$ expected time 
 for a random function $F\colon X\to Y$ such that
$|X| \geq \ell \cdot |Y| / c_N$ for any $1\le c_N \in  o(N^{{1}/({2^\ell - 1})})$.
With the same complexities, it is actually able to find a multiclaw for random functions, which is
harder to find than a multicollision.

\end{abstract}

\end{titlepage}

\section{Introduction}
\label{sec:introduction}
Finding collisions or multicollisions is a fundamental problem in theoretical computer science and one of the most central problems in cryptography. For given finite sets $X$ and $Y$ with $\card{Y} = N$, and a function $F \colon X \to Y$, the $\ell$-collision finding problem is to find a set of $\ell$ distinct inputs $x_1,\dots,x_\ell$ such that $F(x_1) = \dots = F(x_\ell)$. 
Bounding query and time complexities of the $\ell$-collision finding problem is fundamental and has several applications in cryptography. 

\subsubsection*{Applications of multicollisions.}
We often use the lower bound of query complexity (or the upper bound of the success probability) to prove the security of cryptographic schemes. Let us consider a cryptographic scheme based on Pseudo-Random Functions (PRFs). In the security proof, we replace the PRFs with truly random functions (or random oracles) and show the security of the scheme with the random oracles by information-theoretic arguments. In the latter security arguments, we often use \emph{the lower bound} on the 
number of queries required to find multicollisions of random functions. For example, Chang and Nandi~\cite{DBLP:conf/fse/ChangN08} proved the indifferentiability of the chopMD hash function construction; Jaulmes, Joux, and Valette~\cite{DBLP:conf/fse/JaulmesJV02} proved the indistinguishability of RMAC; Hirose~\etal~\cite{DBLP:conf/icisc/HiroseIKOPY10} proved the indifferentiability of the ISO standard lightweight hash function Lesamnta-LW; Naito and Ohta~\cite{DBLP:conf/scn/NaitoO14} improved the indifferentiability of PHOTON and Parazoa hash functions; and Javanovic, Luykx, and Mennink~\cite{JLM14} greatly improved the security lower bounds of authenticated-encryption mode of KeyedSponge. The upper bound of the probability of obtaining multicollisions after $q$ queries plays an important role in their proofs. 

In addition, studying and improving the upper bound for the $\ell$-collision finding problem also help our understanding the complexity of generic attacks. For example, $\ell$-collisions are exploited in the collision attack on the MDC-2 hash function construction by Knudsen~\etal~\cite{DBLP:conf/eurocrypt/KnudsenMRT09}, the preimage attack on the JH hash function by Mendel and Thomsen~\cite{MT08}, the internal state recovery attack on HMAC by Naito~\etal~\cite{DBLP:conf/iwsec/NaitoSWY13}, the key recovery attack on iterated Even-Mansour by Dinur~\etal~\cite{DBLP:conf/asiacrypt/DinurDKS14}, and the key recovery attack on LED block cipher by Nikoli\'{c}, Wang, and Wu~\cite{DBLP:conf/fse/NikolicWW13}.

Furthermore, multicollisions also have applications in protocols. An interesting example is a micro-payment scheme, MicroMint~\cite{RS96a}. Here, a coin is a bit-string the validity of which can be easily checked but hard to produce. In MicroMint, coins are $4$-collisions of a function. If $4$-collisions can be produced quickly, a malicious user can counterfeit coins.
Recently, Bitansky  \etal~\cite{DBLP:conf/stoc/BitanskyKP18} showed that a 3-message zero-knowledge argument against arbitrary polynomial-size non-uniform adversaries can be constructed from multicollision resistant hash functions.
Moreover, Berman \etal~\cite{DBLP:conf/eurocrypt/BermanDRV18} proved that the existence of the multicollision resistant hash functions implies that the existence of constant-round statistically hiding and computationally binding commitment schemes.
Komargodski \etal~\cite{DBLP:conf/eurocrypt/KomargodskiNY18} proved that we can construct some commitment schemes by assuming the existence of a multicollision resistant hash function, instead of assuming the existence of a collision resistant hash function.

\subsubsection*{Existing results on multicollisions in classical settings.}
The problem of finding (multi)collisions has been extensively discussed in the classical setting. Suppose that we can access the function $F$ given as an oracle with \emph{classical} queries; that is, we can send $x \in X$ to the oracle $F$ and obtain $y \in Y$ as $F(x)$.
For a random function $F$, making $q$ queries to $F$ can find a collision of $F$ with a probability bounded by $O(q^2/N)$, which implies that we cannot find collisions of $F$ with high probability 
if $q$ is in $o(N^{1/2})$.
In addition, we obtain a collision with a constant probability by making $O(N^{1/2})$ queries.
Thus, the query complexity of finding a collision of a random function with at least a constant probability is $\Theta(N^{1/2})$.
This can be extended to the general $\ell$-collision cases: 
The bound for finding an $\ell$-collision of a random function is $\Theta(N^{(\ell-1)/\ell})$ (see \cite{STKT08}, for example).

The above argument only focuses on the number of queries.
To implement the $\ell$-collision finding algorithm, the computational cost, $\Time$, and the memory amount, $\Memory$, or their tradeoff should be considered.
The simple method needs to store all the results of the queries.
Hence, it requires that $\Time$ and $\Memory$ are in $\Theta(N^{1/2})$ for collisions, and $\Time$ and $\Memory$ are in $\Theta(N^{(\ell-1)/\ell})$ for $\ell$-collisions.
The collision finding algorithm can be made memoryless by using Floyd's cycle detecting algorithm~\cite{F67-1}.
However, no such memoryless algorithm is known for $\ell$-collisions, and thus the researchers' goal is to achieve a better complexity with respect to $\Time \times \Memory$ or to trade $\Time$ and $\Memory$ for a given $\Time \times \Memory$.

An $\ell$-collision can be found with 
$\Time = O(N)$ 
and $\Memory = \tilde{O}(1)$ by running a brute-force preimage attack $\ell$ times for a fixed target, if the domain is sufficiently large and $\ell$ is constant. Although this method achieves better $\Time \times \Memory$ than the simple method, it cannot trade $\Time$ for $\Memory$. Joux and Lucks~\cite{JL09} discovered the $3$-collision finding algorithm with $\Time \in O(N^{1-\alpha})$ and $\Memory \in O(N^\alpha)$ for $\alpha < 1/3$ by using the parallel collision search technique. Nikoli\'{c} and Sasaki ~\cite{DBLP:conf/asiacrypt/NikolicS16} achieved the same complexity as Joux and Lucks by using an unbalanced meet-in-the-middle attack.

\subsection{Collisions and Multicollisions in Quantum Settings}
Algorithmic speedup using quantum computers has been actively discussed recently. For example, Grover's seminal result~\cite{Gr} attracted cryptographers' attention because of the quantum speedup of database search. Given a function $F \colon X \to \{0,1\}$ such that there exists a unique $x_0 \in X$ that satisfies $F(x_0)=1$, Grover's algorithm finds $x_0$ by making $O\big(\card{X}^{1/2}\big)$ quantum queries.

This paper discusses the complexity of quantum algorithms in \emph{the quantum query model}. In this model, a function $F$ is given as an oracle (\emph{or},  black box), and the complexity of quantum algorithms is measured as the number of quantum queries to $F$. A quantum query model is widely adopted, and previous studies on finding collisions in the quantum setting follow this model~\cite{BHT97,Amb07,Belovs12,Yuen14,Zha15}.

Previous research on finding collisions and multicollisions can be classified 
based on two types of dichotomies.

\begin{description}
\item[Domain size and range size.] The domain size of the function $F \colon X \to Y$ relative to its range size is a sensitive problem for quantum algorithms. Some quantum algorithms aim to find collisions or multicollisions of $F$ with $\card{X} \geq \card{Y}$, while others target $F$ with $\card{X} < \card{Y}$. The former algorithms can be directly applied to find collisions or multicollisions of real \emph{hash functions} such as SHA-3. The latter ones mainly target \emph{database search} rather than hash functions. The latter can also be applied to the case of hash functions, but it generally 
costs much more than the former (in general, the former algorithm cannot be used
for database search).

Hereafter, we use ``H'' and ``D'' to denote the cases with $\card{X} \geq \card{Y}$ and $\card{X} < \card{Y}$, respectively. We note that our goal is to find a new multicollision algorithm that can be applied to hash functions, namely,  the H setting.

\item[Random function and any function.] Both in classical and quantum settings, some algorithms assume the uniform distribution on inputs: they can find a collision  if $F$ is chosen uniformly at random from $\Func(X,Y) := \{f \mid f\colon X \to Y\}$.
Others assume no input distributions:
they can find a collision of \emph{any} function 
$F\in \Func(X,Y)$.
Such algorithms obviously find a collision of a randomly chosen function,
while their complexity
may be much worse than the average-case complexity of algorithms tailor-made for a randomly chosen function
(the complexity of an algorithm for a random function is averaged over the input distribution
and the randomness of algorithms).
In addition, yet other algorithms assume that $F$ is an arbitrary $\ell$-to-$1$ function with $|X|=\ell\cdot |Y|$.
Hereafter, we use ``Rnd'',  ``Arb'', ``Arb$_\ell$'' to denote the case in which $F$ is chosen uniformly at random from $\Func(X,Y)$,  the case in which $F$ is chosen arbitrarily from $\Func(X,Y)$,
and the case in which $F$ is chosen arbitrarily from the set of $\ell$-to-1 functions in $\Func(X,Y)$
with $|X|=\ell\cdot |Y|$.
This paper focuses on the Rnd setting.
\end{description}

In the following, we revisit the existing results on collision-finding and multicollision-finding algorithms in the quantum setting.
\begin{itemize}
\item Brassard, H{\o}yer, and Tapp~\cite{BHT97} proposed a quantum algorithm \algo{BHT}
that finds a 2-collision in the H-Arb$_\ell$ setting.
To be more precise, \algo{BHT} finds a $2$-collision of any $\ell$-to-one function with $O(N^{1/3})$ quantum queries and the memory amount of $O(N^{1/3})$. 
\item Ambainis~\cite{Amb07} studied the element distinctness problem,
that is, the problem of finding an $\ell$-collision in the D-Arb setting.
 The quantum query complexity of the algorithm is  $O(M^{\ell/(\ell+1)})$, where $M$ is the domain size.
\item Belovs~\cite{Belovs12} improved Ambainis' bound
to $O(M^{1 - 2^{\ell-2}/(2^\ell-1)})$.
\item Zhandry~\cite{Zha15} observed that Ambainis' algorithm 
 can be modified so as to find a 2-collision in the H-Rnd setting with $O(N^{1/3})$ quantum queries, when $\card{X}$ is in $\Omega(N^{1/2})$ and $N=\card{Y}$.
\item Yuen~\cite{Yuen14} discussed the application of \algo{BHT} when $\card{X} = \card{Y}$ and the target function $F$ is considered in the Rnd setting. In this case, the quantum query complexity is $O(N^{1/3})$. We omit the details since the discussed case in Yuen's work~\cite{Yuen14} is a subset of Zhandry's extension of Ambainis' algorithm.
\item Regarding the lower bound, the number $O(N^{1/3})$ of queries made by  \algo{BHT} to find a $2$-collision 
in the Arb$_\ell$ setting
was proved to be optimal by 
Refs.~\cite{AS04,Amb05,Kut05}.
Zhandry~\cite{Zha15} proved that the upper bound $O(N^{1/3})$ for 2-collisions 
in the Rnd setting
is optimal
by providing the matching lower bound $\Omega(N^{1/3})$.\footnote{Zhandry showed that any quantum algorithm with $q$ quantum queries finds a $2$-collision with probability at most $O\big(q^3/N\big)$~\cite{Zha15}.} 
Obviously, the lower bound $\Omega(N^{1/3})$ also holds for $\ell > 2$, but no advanced lower bound is known for $\ell > 2$.
\end{itemize}

The classification of the existing algorithms is shown in \autoref{Tbl:classification}. 
\begin{table}[h]
\centering
\caption{Summary of existing quantum algorithms for finding (multi)collisions.}
\label{Tbl:classification}
\begin{tabular}{l@{\hspace{1.5em}}l@{\hspace{1.5em}}l}
\toprule
& Random function ``Rnd'' & Arbitrary function ``Arb'' \\
\midrule
Database ``D''& Zhandry + Ambainis (2-col)& Ambainis ($\ell$-col) \\
&  {\bf Ours} ($\ell$-col) & Belovs ($\ell$-col) \\
\midrule
Hash ``H''& Zhandry + Ambainis (2-col)& BHT ($2$-col) [Arb$_\ell$] \\ 
& Yuen ($2$-col) & Ambainis ($\ell$-col) \\ 
& {\bf Ours} ($\ell$-col)  & Belovs ($\ell$-col) \\
\bottomrule
\end{tabular}
\end{table}

As mentioned earlier, Ambainis' algorithm~\cite{Amb07} and its improvement by Belovs~\cite{Belovs12} originally focused on the database search, but they can also be applied to the hash function setting. However, all the other approaches for the hash function setting only analyze $2$-collisions. Hence, we can conclude that no quantum algorithm exists that is optimized for finding $\ell$-collisions for hash functions.

\subsection{Our Contributions}
Previous algorithms for $\ell$-collision finding in the D-Arb  settings
can be directly applied to the case of  random hash functions.
However, the latter case has not been sufficiently considered especially for the case of general $\ell$.
This motivates us to provide a systematization of knowledge about existing quantum algorithms. Namely, we, for the first time in this field, provide the state of the art of the complexity of finding $\ell$-collisions against hash functions with a direct application, trivial extension, and simple combination of existing results. 
This state of the art sheds light on the problems that require further investigation. 

For the second but main contribution of this paper, we present a new quantum algorithm to find an $\ell$-collision of a hash function chosen at random.

Our contributions in each part are detailed below.

\subsubsection*{Systematization of knowledge (combination of existing results).}
\begin{itemize}
\item Our first observation is that, when $F$ is a random function and $\card{X} = \ell\cdot  \card{Y}$ for any integer constant $\ell$, the query complexity of the $\ell$-collision finding problem is lowered to $O(N^{1/2})$ by simply applying Grover's algorithm. Hence, any meaningful generic attack in the quantum setting must achieve the query complexity in $o(N^{1/2})$. Intuitively, a preimage of the hash value can be generated with $O(N^{1/2})$ queries in the quantum setting and an $\ell$-collision is produced by generating $\ell$ preimages. This corresponds to the $O(N)$ upper bound on the classical complexity (note that this upper bound is for the Rnd setting and does not hold for the Arb setting).
\item The above observation is quite straightforward but useful to measure the effect of other attacks. For example, Ambainis' $\ell$-collision search for database \cite{Amb07} can be 
applied directly
to the hash function case with $O(M^{\ell/(\ell+1)})$ complexity where $M$ is the domain size. However, this cannot be below $O(N^{1/2})$ for any $\ell\ge 2$ if $M \geq N$. The same applies to the improvement by Belovs~\cite{Belovs12}. Those direct applications of the algorithms can be meaningful only in the Arb setting.
\item Zhandry~\cite{Zha15} discussed the application of Ambainis' $\ell$-collision search in H-Rnd and D-Rnd only for $\ell =2$, although it can trivially be extended to the case of $\ell > 2$
(by sampling an $N^{(\ell-1)/\ell}$-size subset of $X$
and applying Ambainis's algorithm to the subset). 
However, the complexity obtained by extending the idea to $\ell =3$ already reaches $O(N^{1/2})$. Thus, Zhandry's idea only works for $\ell =2$.
\item Zhandry~\cite{Zha15} considered Ambainis' $\ell$-collision search rather than Belovs' improvement~\cite{Belovs12}. If we consider Zhandry + Belovs, the complexity in H-Rnd for $\ell =3$ becomes $O(N^{10/21})$, which is faster than the simple application of Grover's algorithm. Thus, it is a meaningful generic attack. For $\ell \ge 4$, the complexity of Zhandry + Belovs reaches $O(N^{1/2})$.
\item In summary, for the Rnd setting, the tight algorithm with $O(N^{1/3})$ complexity exists for $\ell =2$. There is a better generic attack than the simple application of Grover's algorithm for $\ell =3$. 
For $\ell \ge 4$, there is no known algorithm better than the application of Grover's algorithm. 
\end{itemize}

\subsubsection*{New quantum multicollision-finding algorithm.}
Given the above state of the art, our main contribution is a new $\ell$-collision finding algorithm.
For this, we first provide 
an efficient quantum algorithm 
for a more general problem: the 
$\ell$-claw finding problem,
where 
an \emph{$\ell$-claw} for $\ell$ functions $f_i\colon X_i\rightarrow Y$ for $i\in [\ell]$ is a tuple $(x_1,\dots,x_\ell,y) \in X_1 \times \cdots \times X_\ell \times Y$ such that 
$f_i(x_i)=y$ for all $i\in [\ell]$.
\begin{theorem}[\bf Multiclaw-finding: Informal]\label{thm:claw_informal}
Let $N$ be a sufficiently large positive integer, and let $c_N$
be any fixed real satisfying $1\le c_N\in o(N^{{1}/({2^\ell-1})})$.
Suppose that, for each $i\in [\ell]$,
function $f_i\colon X_i\to Y$
is chosen uniformly at random
from the set of all functions
from a set $X_i$ to a set $Y$,
where $\card{Y}=N$ and $\card{X_i}\ge N/c_N$
for each $i\in [\ell]$.
Then, there exists a quantum algorithm that, for $\ell$ functions $f_i \ (i\in [\ell])$ given as an oracle,
finds an $\ell$-claw with probability at least some constant,
where the probability is taken over both the inherent randomness of the algorithm
and the randomness of choices of the functions $f_i$.
Moreover, the algorithm makes at most
$\mathsf{Qlimit} := O\big( c_N \cdot  N^{(2^{\ell-1}-1)/(2^{\ell}-1})\big)$ 
quantum queries,
and runs
in $\tilde{O}(\mathsf{Qlimit})$ time 
on $\tilde{O}(\mathsf{Qlimit})$ qubits
for every possible function $f_i$,
where $\tilde{O}(\cdot)$ suppresses a $\log N$ factor.
\end{theorem}

Via a simple reduction to the $\ell$-claw finding problem (see \autoref{lem:clawtocoll}), we provide
an $\ell$-collision finding algorithm as stated below.

\begin{corollary}[\bf Multicollision-finding: Informal]\label{thm:main_informal}
Let $N$ be a sufficiently large positive integer, and let $c_N$
be any fixed real satisfying $1\le c_N\in o(N^{{1}/({2^\ell-1})})$.
Suppose that a function $F$ is chosen uniformly at random
from the set of all functions from a set $X$ to a set $Y$,
where $\card{Y}=N$ and $\card{X}\ge \ell \cdot N/c_N$.
Then, there exists a quantum algorithm that, for function $F$ given as an oracle,
finds an $\ell$-collision with  probability at least some constant,
where the probability is taken over both the inherent randomness of the algorithm
and the randomness of choices of the function $F$.
Moreover, the algorithm makes at most
$\mathsf{Qlimit} := O\big( c_N \cdot  N^{(2^{\ell-1}-1)/(2^{\ell}-1})\big)$ 
quantum queries,
and runs
in $\tilde{O}(\mathsf{Qlimit})$ time 
on $\tilde{O}(\mathsf{Qlimit})$ qubits
for every possible function $F$,
where $\tilde{O}(\cdot)$ suppresses a $\log N$ factor.
\end{corollary}

We first provide some remarks on the complexity analysis.
When dealing with random functions,
we always consider average case complexity in analyzing (quantum) algorithms, following standard conventions in cryptology.
More concretely, when we discuss average (or expectation) value of query/time/space complexity and success probability of a quantum algorithm, we take the average over both of the randomness involved in the quantum algorithm (termed ``inherent randomness of the algorithm'' in the theorem) and the randomness of problem instances (i.e., the randomness required for choosing input functions).
However, if we accept additional constant errors,  it follows from Markov's inequality that even if we consider the worst case over both of the randomness, the orders for the quantum query/time/space complexity of the algorithm is the same as that for the average case complexity.
In fact, the quantum query/time/space complexity given in the above theorem and corollary
has an upper limit over all possible functions.

In the special case where $c_N$ is a constant, our algorithm 
in \autoref{thm:main_informal}
can find an $\ell$-collision of a random function with $O\big(N^{({2^{\ell-1}-1})/({ 2^{\ell}-1})}\big)$ quantum queries.
We thus achieve a speedup compared with the simple upper bound of $O(N^{1/2})$ for any $\ell$. 
Our upper bound matches the lower bound of $\Omega(N^{1/3})$ for $\ell =2$ and is better than the upper bound $O(N^{10/21})$ of Zhandry + Belovs for $\ell =3$.
In addition, our bound improves the simple bound of $O(N^{1/2})$ for $\ell\ge 4$ for the first time.
The complexity of our algorithm for small constants $\ell$ is shown in \autoref{tbl:numbers}. The complexities are compared in \autoref{fig:Known-Bounds-Rnd}.

Unlike other algorithms for the Arb setting, our algorithm asymptotically approaches  $O(N^{1/2})$ as $\ell$ increases.
The previous results by Ambainis~\cite{Amb07} asymptotically approaches to $O(M)$, and Belovs~\cite{Belovs12} asymptotically approaches to $O(M^{3/4})$, respectively, where $M = \card{X}$.
The complexities are compared in~\autoref{fig:Known-Bounds-Arb} for $M = \ell \cdot N$.

The main idea of our new algorithm is very simple: We just extend the strategy of the $\BHT$ algorithm, which first makes a list $L_1$ of many $1$-collision (i.e., many elements with distinct images), and then extend $L_1$ to a $2$-collision with the Grover search.
By extending this strategy, to find an $\ell$-claw, we first make a list $L_1$ of many $1$-claws.
Then, we make a list $L_i$ of many $i$-claws from $L_{i-1}$ by iteratively applying the Grover search for $2 \leq i \leq \ell-1$  and finally extend $L_{\ell-1}$ to an $\ell$-claw by using the Grover search again.
By optimizing the sizes of the lists, we obtain the quantum query complexity $O(c_N\cdot N^{(2^{\ell-1}-1)/(2^\ell-1)})$. The analysis is not very difficult if the input function $f_i:X_i\to Y$ is one-to-one
or even random with $\card{X_i}\ge \card{Y}$ for every $i\in [\ell]$.
However, we need elaborate analyses in optimizing list sizes to deal with
random functions whose domain sizes may be much smaller than their range size,
since the query complexity heavily depends on the amount of the overlap between
the images of the random functions.
We provide a rigorous complexity evaluation of our algorithm, which is another main focus of this paper.
Our evaluation suggests that our algorithm finds a $2$-collision of SHA3-512 with $2^{181}$ quantum queries and finds a $3$-collision with $2^{230}$ quantum queries.

\begin{figure}[h]
\centering
\begin{tikzpicture}[scale=.5,
inner sep=1pt,
ours/.style={circle,draw=black,fill=black,thin,radius=2pt}]
\draw[->,thin] (-0.5,0) -- (8.5,0);
\draw[->,thin] (0,-1/2) -- (0,6.5);
\node[anchor=south] at (0,6.7) {$\log_NQ$};
\node[anchor=west] at (9,0) {$\ell$};
\node[gray!80,anchor=east] at (-0.5,0) {$1/3$};
\draw[gray!80,very thin] (-1/2,6) node[anchor=east] {$1/2$} -- (8.5,6); 
\node[anchor=west] at (8.5,6) {$\longleftarrow$ Trivial upper bound};
\foreach \x in {2,...,8} {\draw[gray!80,very thin] (\x,-1/2) node[anchor=north] {$\x$} -- (\x,6.5);}
\draw (2,36 * 1/3 - 12)			circle [radius=5pt];
\draw (3,36 * 10/21 - 12)			circle [radius=5pt];
\draw (4,36 * 1/2 - 12)			circle [radius=5pt];
\draw (5,36 * 1/2 - 12)			circle [radius=5pt];
\draw (6,36 * 1/2 - 12)		circle [radius=5pt];
\draw (7,36 * 1/2 - 12)		circle [radius=5pt];
\draw (8,6)		circle [radius=5pt];
\filldraw [ours] (2,36 * 1/3 - 12)			circle [radius=3pt];
\filldraw [ours] (3,36 * 3/7 - 12)			circle [radius=3pt];
\filldraw [ours] (4,36 * 7/15 - 12)			circle [radius=3pt];
\filldraw [ours] (5,36 * 15/31 - 12)			circle [radius=3pt];
\filldraw [ours] (6,36 * 31/63 - 12)		circle [radius=3pt];
\filldraw [ours] (7,36 * 63/127 - 12)		circle [radius=3pt];
\filldraw [ours] (8,36 * 127/255 - 12)		circle [radius=3pt];
\draw (10.5,2.5) circle [radius=5pt];
\node [anchor=west] at (11.1,2.5) {: Known bound for $\ell = 2$};
\node [anchor=west] at (11.7,1.7) {and our observation for $\ell \geq 3$};
\filldraw [ours] (10.5,0.9) circle [radius=3pt];
\node [anchor=west] at (11.1,0.8) {: Our algorithm};
\end{tikzpicture}
\caption{Quantum query complexity for finding an $\ell$-collision for a random function: The label of the vertical axis, $\log_NQ$,  denotes the logarithm of  the number of queries to the base $N$, where $N$ is the size of the range of the function.}
\label{fig:Known-Bounds-Rnd}
\end{figure}
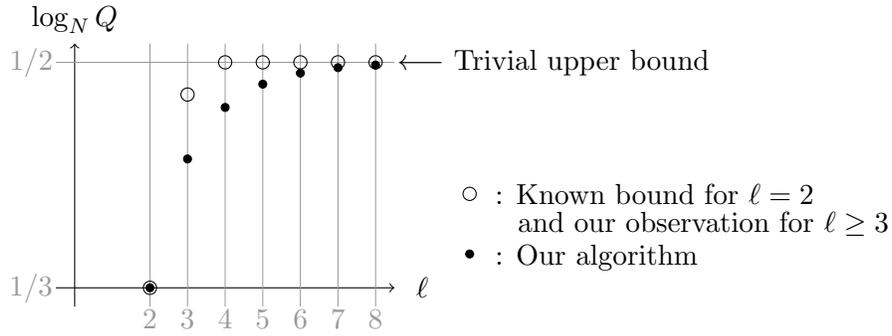

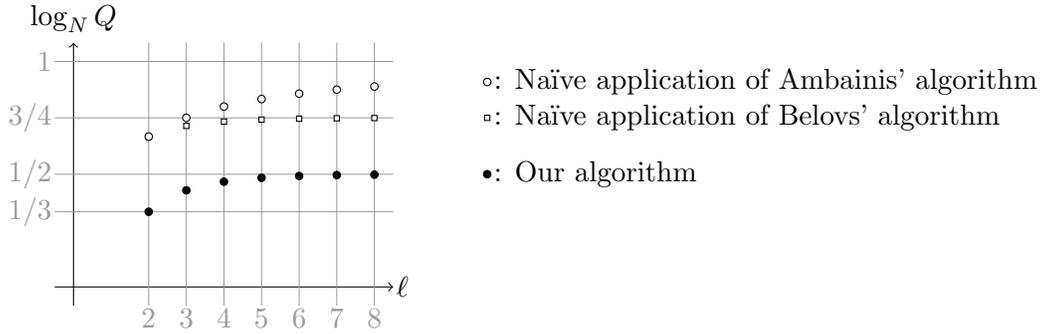
\begin{figure}[h]
\centering
\begin{tikzpicture}[scale=.5,
inner sep=1pt,
q-arb-bel/.style={rectangle,draw=black,fill=white,thin},
ours/.style={circle,draw=black,fill=black,thin},%
q-arb-amb/.style={circle,draw=black,fill=white,thin}%
]
\draw[->,thin] (-0.5,0) -- (8.5,0);
\draw[->,thin] (0,-1/2) -- (0,6.5);
\node[anchor=south] at (0,6.7) {$\log_NQ$};
\node[anchor=west] at (8.5,0) {$\ell$};
\draw[gray!80,very thin] (-1/2,9/2) node[anchor=east] {$3/4$} -- (8.5,9/2);
\draw[gray!80,very thin] (-1/2,2) node[anchor=east] {$1/3$} -- (8.5,2);
\draw[gray!80,very thin] (-1/2,3) node[anchor=east] {$1/2$} -- (8.5,3);
\draw[gray!80,very thin] (-1/2,6) node[anchor=east] {$1$} -- (8.5,6);
\foreach \x in {2,...,8} {\draw[gray!80,very thin] (\x,-1/2) node[anchor=north] {$\x$} -- (\x,6.5);}
\node [q-arb-bel] at (2,4) {};
\node [q-arb-bel] at (3,30/7) {}; 
\node [q-arb-bel] at (4,22/5) {};
\node [q-arb-bel] at (5,138/31) {};
\node [q-arb-bel] at (6,94/21) {};
\node [q-arb-bel] at (7,570/127) {};
\node [q-arb-bel] at (8,382/85) {};
\node [q-arb-amb] at (2,6*2/3) {};
\node [q-arb-amb] at (3,6*3/4) {}; 
\node [q-arb-amb] at (4,6*4/5) {};
\node [q-arb-amb] at (5,6*5/6) {};
\node [q-arb-amb] at (6,6*6/7) {};
\node [q-arb-amb] at (7,6*7/8) {};
\node [q-arb-amb] at (8,6*8/9) {};
\node [ours] at (2,6 * 1/3) {};
\node [ours] at (3,6 * 3/7) {};
\node [ours] at (4,6 * 7/15) {};
\node [ours] at (5,6 * 15/31) {};
\node [ours] at (6,6 * 31/63) {};
\node [ours] at (7,6 * 63/127) {};
\node [ours] at (8,6 * 127/255) {};
\node [q-arb-amb] at (11, 6 * 10/11) {};
\node [anchor=west] at (11.1, 6 * 10/11) {: Na\"ive application of Ambainis' algorithm};
\node [q-arb-bel] at (11, 6 * 3/4) {};
\node [anchor=west] at (11.1, 6 * 3/4) {: Na\"ive application of Belovs' algorithm};
\node [ours] at (11,6*1/2) {};
\node [anchor=west] at (11.1,6*1/2) {: Our algorithm};
\end{tikzpicture}
\caption{Quantum query complexity for finding an $\ell$-collision in H-Rnd setting: The label of the vertical axis, $\log_NQ$,  denotes the logarithm of  the number of queries to the base $N$, where $N$ is the size of the range of the function.}
\label{fig:Known-Bounds-Arb}
\end{figure}

\begin{table}[htb]
\begin{center}
\caption{Quantum query complexities of our algorithm for finding an $\ell$-collision for a random function for $\ell =2,\dots, 8$: Each fraction denotes the logarithm of the number of queries to the base $N$,
where $N$ is the size of codomain of the function,
and its approximate value is shown just below it.
The query complexity asymptotically approaches $1/2$ as $\ell$ increases.}
\label{tbl:numbers}
\renewcommand{\arraystretch}{2}
\begin{tabular}{c | c@{}c@{}c@{}c@{}c@{}c@{}c }
\toprule
$\ell$ & \makebox[1cm]{2}  & \makebox[1cm]{3} & \makebox[1cm]{4} & \makebox[1cm]{5} & \makebox[1cm]{6} & \makebox[1cm]{7} & \makebox[1cm]{8} \\
\midrule
{\large $\frac{2^{\ell-1}-1}{2^{\ell}-1}$} & {\Large $\frac{1}{3}$}& {\Large $\frac{3}{7}$} & {\Large $\frac{7}{15}$} & {\Large $\frac{15}{31}$} & {\Large $\frac{31}{63}$} & {\Large $\frac{63}{127}$} & {\Large $\frac{127}{255}$}  \\
 & {\small $0.3333..\ $}& {\small $0.4285..\ $} & {\small $0.4666..\ $} & {\small $0.4838..\ $} & {\small $0.4920..\ $} & {\small $0.4960..\ $} & {\small $0.4980..\ $}  \\ \bottomrule
\end{tabular}
\end{center}
\end{table}

\subsection{Updates from the Conference Versions}
Two preliminary versions of this work were presented at Asiacrypt 2017~\cite{DBLP:conf/asiacrypt/HosoyamadaSX17} and PQCrypto 2019~\cite{PQC19}.
The systematization of knowledge in this paper is from the former version,
whereas our main algorithm is from the latter version.
In this paper, we newly provide time/space complexity analysis of the algorithm.
The algorithm in the former version is not presented in this paper, since the algorithm in the latter version (this paper's main algorithm) improves upon the former one in terms of query/time complexity.
In terms of space complexity, however, the algorithm in the former version has a certain advantage, which will be discussed
in Section~\ref{subsec:A_Recursive_Multicollision-Finding_Algorithm}.
Our main algorithm involves intermediate measurements, which can be postponed to the end of the algorithm
at the sacrifice of additional space.
A discussion on these intermediate measurements is newly added in this paper.

\subsection{Paper Outline.}
The remaining of this paper is organized as follows.
\hnote{\autoref{sec:prelim} describes notations, definitions and a version of quantum search algorithm,
and
\autoref{sec:sok} systematizes the knowledge on previous quantum algorithms related to the multicollision-finding problem.
\autoref{sec:alg} then provides our quantum algorithm for finding a multiclaw and its complexity analysis.
\autoref{sec:discussions} explains how to run our algorithm without any intermediate measurements,
and discusses how to save work space at the sacrifice of query/time complexity.
\autoref{sec:conclusion} concludes this paper and presents some open problems.}

\subsection{Concurrent Work.}
Very recently, Liu and Zhandry~\cite{DBLP:journals/iacr/LiuZ18} showed that
 for every integer constant $\ell\ge 2$, $\Theta\big(N^{\frac{1}{2}(1-{1}/(2^\ell-1))}\big)$ quantum queries are both necessary and sufficient on average to find an $\ell$-collision with probability at least some constant for a random function. The comparisons are summarized as follows:
\begin{itemize}
\item Liu and Zhandry consider only the case where that $|X| \geq \ell |Y|$, where $X$ is the domain and $Y$ is the range, while we treat a more general case where $|X| \geq {\ell}
\cdot |Y|/{c_N}$ holds for any positive value $c_N\ge 1$ which is in $o\big(N^{{1}/({2^\ell-1})}\big)$. 
When $c_N$ is a constant, our bound $O\big(c_N \cdot N^{(2^{\ell-1}-1)/(2^\ell-1)}\big)$ matches theirs $O\big(N^{\frac{1}{2}(1-{1}/({2^\ell-1}))}\big)$.
\item They give a lower bound, which matches with their upper bound for $|X| \geq \ell |Y|$
 and ours for $|X| \geq {\ell}|Y|/{c_N} $ with any constant $c_N\ge 1$.
\end{itemize}
We finally note that our $\ell$-collision finding algorithm
 for the case $|X| \geq \ell |Y|$ with query complexity $O(N^{\frac{1}{2}\left(1-{1}/(2^\ell-1)\right)})$
 was reported in the Rump Session at Asiacrypt 2017, while Liu and Zhandry's (technical) paper appeared in 2018.

\section{Preliminaries}
\label{sec:prelim}
For a positive integer $K$, let $[K]$ denote the set $\{1,\dots,K\}$.
Unless otherwise noted, all sets are non-empty and finite.
For \tnote{sets} $X$ and $Y$, $\Func(X,Y)$ denotes the set of functions from $X$ to $Y$.
For each $f \in \Func(X,Y)$, we denote the set $\{f(x) \mid x \in X\}$ by $\mathrm{Im}(f)$.
For a {set} $X$, let $U(X)$ denote the uniform distribution over $X$.
For a distribution $\mathcal{D}$ on a set $X$,  we mean by $x \sim \mathcal{D}$ that  $x$ is a random variable that takes a value drawn from $X$ according to $\mathcal{D}$.
When we say that an oracle of a function $f \colon X \rightarrow Y$ is given, we consider the situation that each element $x\in X$ ($y\in Y$) is encoded as a distinct binary string of length $\ceil{\log_2 |X|}$ ($\ceil{\log_2 |Y|}$), and the oracle gate $O_f \colon \ket{x,z} \mapsto \ket{x,z \oplus f(x)}$ is available,
where $z\in 2^{\ceil{\log_2 |Y|}}$.
We assume that $\ell$ is an integer constant throughout this paper.

An \emph{$\ell$-collision} for a function $f\colon X \rightarrow Y$ is a tuple of elements $(x_1,\dots,x_\ell,y)$ in $X^\ell \times Y$ such that $f(x_i)=f(x_j)=y$ and $x_i \neq x_j$ for all $i,j \in [\ell]$ with $i\neq j$.
An $\ell$-collision is simply called a \emph{collision} for $l=2$, and called a \emph{multicollision} for $l \geq 3$.
Moreover, an \emph{$\ell$-claw} for $\ell$ functions $f_i\colon X_i\rightarrow Y$ for $i\in [\ell]$ is a tuple $(x_1,\dots,x_\ell,y) \in X_1 \times \cdots \times X_\ell \times Y$ such that 
$f_i(x_i)=y$ for all $i\in [\ell]$.
An $\ell$-claw is simply called a \emph{claw} for $l=2$, and called a \emph{multiclaw} for $l \geq 3$.

The problems of finding multicollisions or multiclaws are often studied in the \tnote{contexts} of both cryptography and quantum computation, but the problem \tnote{settings} of interest change depending on the \tnote{contexts}.
In the context of quantum computation, most problems are studied 
in terms of the \emph{worst-case} complexity over all possible inputs.
More concretely, let $\textsc{cost}^{\cal A}_{1-\varepsilon}(I)$ be the worst-case complexity,
over the inherent randomness in an algorithm ${\cal A}$, required for ${\cal A}$ to output a correct answer with error probability
at most $\varepsilon>0$ for input $I$.
Then, we are interested in finding ${\cal A}$
that minimizes $\max_I\textsc{cost}^{\cal A}_{1-\varepsilon}(I)$  for a prespecified small $\varepsilon$.
On the other hand, most problems in cryptography are studied in terms of  the \emph{average-case}
complexity over input distribution,
since the randomness of input is one of the most crucial notions in cryptography.
More concretely, 
let $\overline{\textsc{cost}}^{\cal B}_{1-\epsilon}({\cal D})$
be the average complexity of an algorithm ${\cal B}$,
over
the inherent randomness in ${\cal B}$
and the randomness of input subject to distribution ${\cal D}$,
required  for ${\cal B}$ to output a correct answer with probability 
at least $1-\epsilon$, where the probability is taken over both of the randomness.
Then, we are interested in  finding ${\cal B}$ that minimizes
$\overline{\textsc{cost}}^{\cal B}_{1-\epsilon}({\cal D})$.
By Markov's inequality, it is possible to upper-bound the worst-case complexity
at the sacrifice of some additional error.
For instance, consider query complexity as complexity measure.
If we restrict the maximum allowable number of queries that ${\cal B}$ makes
to $k\cdot \overline{\textsc{cost}}^{\cal B}_{1-\epsilon}({\cal D})$, then
${\cal B}$ outputs a correct answer
with probability at least $1-\epsilon-1/k$ 
by making 
at most $k\cdot \overline{\textsc{cost}}^{\cal B}_{1-\epsilon}({\cal D})$ queries 
for \emph{every} inputs.
In fact, the quantum query complexity given in the main theorem (\autoref{thm:main})
is stated in this form.

This paper focuses on the settings of interest in the context of cryptography.
Formally, our goal is to solve the following two problems.
\begin{problem}[Multicollision-finding problem, average case]\label{prob:mcoll}
Let $\ell\ge 2$ be a positive integer constant, and $X,Y$ denote non-empty finite sets.
Suppose that \tnote{a function} $F \colon X \rightarrow Y$ is chosen uniformly at random and \tnote{given} as a quantum oracle.
Then, find an $\ell$-collision for $F$.
\end{problem}
\begin{problem}[Multiclaw-finding problem, average case]\label{prob:mclaw}
Let $\ell\ge 2$ be a positive integer constant, and $X_1, \dots, X_\ell,Y$ denote \tnote{non-empty} finite sets.
Suppose that $\ell$ functions $f_i \colon X_i \rightarrow Y\  (i\in [\ell])$ are chosen independently and uniformly at random, and \tnote{given} as quantum oracles.
Then, find an $\ell$-claw for $f_1,\dots,f_\ell$.
\end{problem}

Roughly speaking, \autoref{prob:mcoll} is easier to solve than \autoref{prob:mclaw}.
Suppose \tnote{that} $F \colon X \rightarrow Y$ is a function, and we want to find an $\ell$-collision for $F$.
Let $X_1,\dots,X_\ell$ be disjoint subsets of $X$ such that 
$\bigcup_i X_i = X$.
If $(x_1,\dots,x_\ell,y)$ is an $\ell$-claw for $F|_{X_1}, \dots, F|_{X_\ell}$, then it is obviously an $\ell$-collision for $F$.
In general, an algorithm for finding an \tnote{$\ell$-claw} can be converted into one for finding an $\ell$-collision.
To be precise, the following lemma holds.
\begin{lemma}\label{lem:clawtocoll}
Let $X,Y$ be non-empty \tnote{finite} sets, and $X_1,\dots,X_\ell$ be (almost-)equal-sized disjoint subsets of $X$ such that 
$\bigcup_i X_i = X$.
If there exists a quantum algorithm $\A$ that solves \autoref{prob:mclaw} with probability at least $p$
for the sets $X_1,\dots,X_\ell,Y$ by making at most $q$ quantum queries, then there exists a quantum algorithm $\B$ that solves \autoref{prob:mcoll} with probability at least $p$ for the sets $X,Y$ by making at most $q$ quantum queries.
\end{lemma}

How to measure the size of a problem also changes depending on which context we are in.
In the context of cryptography, the problem size is often regarded as the size of the range of input functions rather than the sizes of their domains, since the domains of cryptographic functions such as hash functions are much larger than their ranges.
Hence, we regard the range size $|Y|$ as the size of \autoref{prob:mcoll} (and~\autoref{prob:mclaw}) when we analyze the complexity of quantum algorithms.

\subsection{The Grover Search and Its Generalization}
As a main tool for developing quantum algorithms, we use the quantum database search algorithm that was originally developed by Grover~\cite{Gr} and later generalized by Boyer, Brassard, H{\o}yer, and Tapp~\cite{boyer1998tight} for mult-target cases.

\begin{theorem}[\cite{boyer1998tight}]\label{thm:BBHT}
Let $X$ be a non-empty finite set and $f \colon  X \rightarrow \{0,1\}$ be a function such that $t / |X| < 17/81$, where $t = |f^{-1}(1)| $.
Then, there exists a quantum search algorithm $\BBHT$ that, for given $f$ as an oracle, finds $x$ such that $f(x)=1$ with an expected number of quantum queries at most 
$
\frac{4|X|}{\sqrt{(|X|-t)t}} \leq \frac{9}{2} \cdot \sqrt{\frac{|X|}{t}}.
$
If $f^{-1}(1) = \emptyset$, then $\BBHT$ runs forever.
\end{theorem}
We also use this theorem  in the following form.
\begin{corollary}\label{cor:MTPS}
Let $X,Y$ be non-empty finite sets, $f \colon X \rightarrow Y$ be a function, and $Y' \subset Y$ be a non-empty subset.
Then, there exists a quantum search algorithm $\MTPS$ that, for given $f$ as an oracle, finds $x$ such that $f(x) \in Y'$ with an expected number of quantum queries at most $9\sqrt{5\card{X}/\card{f^{-1}(Y')}}$.
\end{corollary}
\begin{proof}
Let $F_{Y'} \colon \{1,\dots,5\} \times X \rightarrow \{0,1\}$ be the boolean function defined as $F_{Y'}(\alpha,x) = 1$ if and only if $\alpha=1$ and $f(x) \in Y'$.
A quantum circuit that computes $F_{Y'}$ can be implemented with two oracle calls to the oracle of $f$.
Then, run $\BBHT$ on $F_{Y'}$.
Since $\card{\{1,\dots,5\} \times X} = 5\card{X}$ and $\card{F_{Y'}^{-1}(1)} \leq \card{X} \leq 17/81 \cdot \card{\{1,\dots,5\} \times X}$ always hold, the corollary follows from~\autoref{thm:BBHT}. 
\end{proof}

\subsection{Tail Bounds of Probability Distributions}
\begin{theorem}[McDiarmid's Inequality (Theorem 13.7 in \cite{mitzenmacher2017probability})]
\label{th:McDiarmid}
Let $M$ be a positive integer, and $\Phi \colon Y^{\times M} \colon \rightarrow \mathbb{N}$ be a $1$-Lipschitz function.
Let $\{y_i\}_{1 \leq i \leq M}$ be the set of independent random variables that take values in $Y$.
Let $\mu$ denote the expectation value $\E_{y_1,\dots,y_M}\left[ \Phi(y_1,\dots,y_M)  \right]$.
Then
\begin{equation*}
\Pr_{y_1,\dots,y_M}\left[ \Phi(y_1,\dots,y_M) \geq \mu + \lambda \right] \leq 2e^{-2\lambda^2/M}
\end{equation*}
holds.
\end{theorem}

\begin{theorem}[Theorem 1 in~\cite{hush2005concentration}]
\label{th:BoundOnHyperGeometricDistribution}
Let $K=K(n_1,n,m)$ denote the hypergeometric random variable describing the process of counting how many defectives are selected when $n_1$ items are randomly selected without replacement from $n$ items among which there are $m$ defective ones. Let $\lambda \geq 2$.
Then
\begin{equation*}
\Pr\left[K - \E[K] < -\lambda \right] < e^{-2\alpha_{n_1,n,m}(\lambda^2-1)}
\end{equation*}
holds, where
\begin{equation*}
\alpha_{n_1,m,n} = \max\left( \left( \frac{1}{n_1+1} + \frac{1}{n-n_1+1} \right), \left(\frac{1}{m+1} + \frac{1}{n-m+1} \right)\right).
\end{equation*}
\end{theorem}


\section{Systematization of Knowledge on Quantum Multicollision Algorithms} 
\label{sec:sok}
In the classical setting, an $\ell$-collision on a hash function can be found with $O(N^{(\ell-1)/\ell})$ queries.
However, the problem has not received much attention in the quantum setting. This section surveys previous work and integrates the findings of different researchers to make several new observations on this topic.
In what follows, we consider the problem of finding a (multi)collision of a function $F : X \rightarrow Y$.

\subsection{Survey of Previous Work}
\label{sec:survey}
We review the algorithm \algo{BHT}~\cite{BHT97}, since our algorithm explained in \autoref{sec:alg} is its extension. We also survey previous studies, classifying them in two types: element $\ell$-distinctness problem (D-Arb), and collision-finding problem on random functions (D-Rnd and H-Rnd).

\subsubsection{BHT: Collision finding problem on $\ell$-to-one functions.}
\label{subsubsec:BHT_for_collision}
For simplicity, we describe \algo{BHT} only for the case $\ell =2$. Let $X, Y$ be sets that satisfy $\card{X} = 2 \cdot \card{Y} = 2N$, and $F \colon X \to Y$ be a $2$-to-one function.

The basic idea of \algo{BHT} is as follows. First, we choose a parameter $k$ ($k = N^{1/3}$ will turn out to be optimal) and a subset $X' \subset X$ of cardinality $k$. We then make a list $L = \{ (x,F(x)) \}_{x \in X'}$. 
If $L$ contains  pairs $(x,F(x))$ and $(x',F(x'))$ such that $F(x)=F(x')$, then we are done.
Suppose that no such pairs exist.
We then use the $\BBHT$ algorithm to find an element $x \in X\setminus X'$ such that there exists $x_0 \in X'$ that satisfies $(x_0,F(x)) \in L$, i.e., we try to extend a pair $(x_0,F(x_0)) \in L$ to a collision $(\{x,x_0\},F(x_0))$ for a certain $x\in X\setminus X'$. The precise description of \algo{BHT} is as follows.\\

\noindent\textbf{Algorithm $\algo{BHT}(F,k)$.}
\begin{enumerate}
\item Choose an arbitrary subset $X' \subset X$ of cardinality $k$. 
\item Make a list $L = \big\{ \big(x,F(x)\big) \big\}_{x \in X'}$ by querying $x \in X'$ to $F$. 
\item Sort $L$ in accordance with $F(x)$.
\item Check whether $L$ contains a $2$-collision, i.e., there exist pairs $(x,F(x)),(x',F(x'))$ in $L$ such that $x \neq x' $ and $F(x) = F(x')$. If so, output the $2$-collision $(\{x,x'\},F(x))$. Otherwise proceed to the next step.
\item Construct the oracle $H \colon X \to \{0,1\}$ by defining $H(x)=1$ if and only if there exists $x_0 \in X'$ such that $(x_0,F(x)) \in L$ and $x \neq x_0$.
\item Run $\algo{BBHT}(F)$ to find $\tilde{x} \in X$ such that $H(\tilde{x}) = 1$. 
\item Find $x_0 \in X'$ that satisfies $F(\tilde{x}) = F(x_0)$.
Output the $2$-collision $(\{\tilde{x},x_0\},F(x_0))$.
\end{enumerate}

This algorithm makes $k$ quantum queries in Step~2 and $O(\sqrt{N/k})$ quantum queries in Step~6 (in fact, in constructing the list $L$, we need no advantage of quantum calculation, so queries in Step~2 can also be made \emph{classically} if we are allowed to access a classical oracle of $F$). Thus, the total number of quantum queries is $O(k + \sqrt{N/k})$, which is minimized when $k=N^{1/3}$. Brassard~\etal~provide the following theorem~\cite{BHT97}.
\begin{theorem}
[{\cite[Theorem~1]{BHT97}}] \label{bht} Suppose that $X$ and $Y$ are finite sets that satisfy $\card{X} = 2 \cdot \card{Y}=2N$, and $F \colon X \to Y$ is a two-to-one function. For $k\in [N]$,
\algo{BHT} finds a $2$-collision of $F$ with an expected quantum query complexity $O( k + \sqrt{N/k})$
and space complexity $\tilde{O}(k)$. In particular, when $k = N^{1/3}$, \algo{BHT} finds a $2$-collision of $F$ with expected quantum query complexity $O({N}^{1/3})$ and space complexity $\tilde{O}(N^{1/3})$.
\end{theorem}
We also note that the expected time complexity of the algorithm in the above theorem
is the same order as the quantum query complexity up to a logarithmic factor,
assuming that a single oracle access takes a constant time.

\subsubsection{Element $\ell$-distinctness problem ($\ell$-collisions in D-Arb).}
Consider \emph{the element $\ell$-distinctness problem}, in which we are given access to the oracle $F \colon X \to Y$ to find whether there exist distinct $x_1,\dots,x_\ell$ such that $F(x_1) = \dots = F(x_\ell)$, i.e., there exits an $\ell$-collision of $F$. Note that $F$ obviously has an $\ell$-collision if $|X| > (\ell - 1)|Y|$, and the element $\ell$-distinctness problem considers the collision-detecting problem on \emph{database} rather than a hash function.

Ambainis \cite{Amb07} proposed a quantum algorithm based on quantum walks that solves the element $\ell$-distinctness problem with $O\big(\card{X}^{\ell/(\ell+1)}\big)$ queries. His algorithm not only decides whether there exists an $\ell$-collision but also finds an actual $\ell$-collision $(\{x_1, \dots, x_\ell\},y)$. It can thus be applied to collision-finding  in the case of $|X| > (\ell - 1)\card{Y}$. 
Subsequently, Belovs~\cite{Belovs12} improved the query complexity 
to $O\big(\card{X}^{1-2^{\ell-2}/(2^\ell-1)}\big) = o(\card{X}^{3/4})$.%
\footnote{The best known time complexity remained $\tilde{O}(\card{X}^{\ell/(\ell +1)})$~\cite{Amb07}.
Subsequently,  Belovs, Childs, Jeffery, Kothari, and Magniez~\cite{BCJKM13} 
and Jeffery~\cite{Jeffery14}
improved the time complexity for $\ell = 3$ to $\tilde{O}(\card{X}^{5/7})$.
However, the quantum query complexity is still $\tilde{O}(\card{X}^{5/7})$,
which is the same as Belovs's bound~\cite{Belovs12}: $\tilde{O}\big(\card{X}^{1-2^{3-2}/(2^3-1)}\big) = \tilde{O}(\card{X}^{5/7})$.}

Although the algorithms by Ambainis and Belovs can be used to find an $\ell$-collision for $\card{X} > (\ell - 1)\card{Y}$, the complexity increases as the domain size $\card{X}$ increases. These algorithms are inefficient to find collisions of hash functions, since the domain size of cryptographic hash functions could be exponentially larger than the range size.
Thus, we often regard the problem size as 
the range size $\card{Y}$ not the domain size $\card{X}$,
and are  interested in complexity in terms of $\card{Y}$.
For this, we need quantum algorithms dedicated to finding collisions of hash functions.

\subsubsection{Collision finding problem on random functions ($\ell$-collisions in D-Rnd and H-Rnd).}
Among variants of the collision problem, the \emph{collision finding problem on random functions} is the most central problem in the context of cryptography. We introduce algorithms for $\ell =2$ in the following. Assume that $F\sim U(\Func(X,Y))$.

\paragraph{Direct Application of  \algo{BHT}.}
Brassard~\etal~\cite{BHT97} mentioned that, 
if  $\card{X} \geq \ell \card{Y}$,
$\algo{BHT}$ 
finds a collision with quantum query complexity $O(\card{Y}^{1/3})$ with probability at least some constant
($\algo{BHT}$ actually works for an arbitrary function $F\in \Func(X,Y)$ in this large domain case,
if the subset $X'\subset X$ is taken uniformly at random).
Yuen~\cite{Yuen14} showed that if $\card{X} = \card{Y}$,
 then 
$\algo{BHT}$ 
finds a collision with quantum query complexity $O(\card{Y}^{1/3})$ with probability at least some constant.

\paragraph{Zhandry's algorithm.} Zhandry~\cite{Zha15} proposed a quantum algorithm finding a collision with $O(|Y|^{1/3})$-quantum queries as long as $\card{X} = \Omega(|Y|^{1/2})$.
 This relaxes the restrictions imposed on domain sizes in Refs.~\cite{BHT97,Yuen14}.
His algorithm is as follows:
Choose a random subset $X' \subset X$ of size $\card{Y}^{1/2}$
and then invoke Ambainis' algorithm~\cite{Amb07} for $F|_{X'} \colon X' \to Y$ to obtain a collision.
The collision exists if $F$ is random because of the birthday paradox, and the query complexity is $O\big(\card{X'}^{2/3}) = O\big( (\card{Y}^{1/2})^{2/3} \big) = O(\card{Y}^{1/3})$. 

\subsection{Observation on Upper Bounds}
This section provides some observation 
for a function $F\sim U(\Func(X,Y))$:
\begin{enumerate}
\item 
If $\card{X}\ge \card{Y}=N$,
the quantum query complexity for finding an $\ell$-collision of $F$
is trivially upper-bounded by $O(N^{1/2})$.
\item If $\card{X} \geq (3!)^{1/3} |Y|^{(3 - 1)/3}$,
the quantum query complexity for finding a $3$-collision of $F$ is upper-bounded by $O(N^{10/21})$.
\end{enumerate}
Observation~1 is obtained by applying a generalized version of Grover's algorithm, and Observation~2 is obtained by combining the idea of Zhandry~\cite{Zha15} with the result of Belovs~\cite{Belovs12}. The details are as follows.

\subsubsection{Trivial upper-bound for finding $\ell$-collisions.} 
\label{sec:mult-grov}
The basic observation is that a constant fraction of elements of $X$ are members of $\ell$-collisions,
if $\card{X}\ge \card{Y}$.
Based on this, one can easily derive the classical upper bound $O(N)$
on the query complexity for finding an $\ell$-collision:
choose an element $x_1 \in X$ uniformly at random,%
\footnote{It suffices to take $x_1$ deterministically, if the algorithm need not be repeated for amplification.}
and then perform exhaustive search to find $x_i\ (x_i\neq x_1)$ such that $F(x_i) = F(x_1)$
for $i = 2,\dots,\ell$.
One can amplify the success probability to, say, 2/3 by just repeating this procedure constant times.
In the quantum setting, we can replace the exhaustive search with \algo{BBHT}. We call this algorithm $\algo{Multi\mathchar`-Grover}$, described as follows:\\

\noindent\textbf{Algorithm $\algo{Multi\mathchar`-Grover}(F)$}
\begin{enumerate}
\item Choose an element $x_1 \in X$ uniformly at random and set $L = \{x_1\}$. 
\item While $\card{L} < \ell$, do: 
\begin{enumerate}
\item Invoke $\algo{BBHT}(F)$ to find $x \in X$ such that $H(x) =1$, where we implement $H \colon X \to \{0,1\}$ as $H(x) = 1$ if and only if $F(x) = F(x_1)$ and $x\neq x_1$. 
\item 
$L \leftarrow L \cup \{x\}$.
\end{enumerate}
\item Output $(L,F(x_1))$ as an $\ell$-collision. 
\end{enumerate}

Roughly speaking, each step in the loop requires $O(N^{1/2})$ queries to find $x_i$. Thus, the total query complexity is $O(N^{1/2})$ for any constant $\ell$. Therefore, to achieve a meaningful improvement, we need to find an $\ell$-collision with $o(N^{1/2})$ quantum queries. 

We note that the lower bound of $2$-collisions given in~\cite{Zha15} also applies to the case of multicollisions.
Hence, the complexity of optimal multicollision-finding algorithms must be something between the tight lower bound $\Omega(N^{1/3})$ for $2$-collisions and 
the trivial upper bound $O(N^{1/2})$.
This corresponds to something between the birthday bound and the preimage bound in the classical setting.

\subsubsection{Extension of element $\ell$-distinctness to $\ell$-collision.}
We observe that algorithms for the $\ell$-distinctness problem can be used to find $\ell$-collisions of a random function $F \colon X \to Y$ by extending Zhandry's idea. Let $X,Y$ be finite sets with $\card{Y} = N$ and $\card{X} \geq (\ell!)^{1/\ell} N^{(\ell - 1)/\ell}$. 
\begin{enumerate}
\item Choose a random subset $X' \subset X$ of size $(\ell!)^{1/\ell} N^{(\ell - 1)/\ell}$.
\item Invoke Belovs' algorithm~\cite{Belovs12} for $F|_{X'} \colon X' \to Y$ to obtain an $\ell$-collision.
\end{enumerate}
According to the precise analysis by Suzuki, Tonien, Kurosawa, and Toyota~\cite{STKT08}, $F|_{X'}$ has an $\ell$-collision with probability approximately $1/2$. Thus, we observe that Belovs' algorithm can find an $\ell$-collision of $F|_{X'}$ with quantum query complexity $O\big( (N^{1-2^{\ell-2}/(2^\ell-1)})^{(\ell - 1)/\ell} \big)$.\footnote{The approach is not improved by picking a smaller random subset. For example, consider finding a $3$-collision of random function $F$. If we pick a smaller random subset $X'$ of size $N^b$ with $b < 2/3,$, then the probability that $X'$ contains a $3$-collision is roughly $N^{3b}/N^2$. Thus, we need to iterate Belovs' algorithm $N^2/N^{3b}$ times (or use with it the quantum amplitude amplification), where each iteration makes $N^{5b/7}$ queries. Therefore, the total number of queries is $N^{(14-16b)/7} > N^{10/21}$ for $b < 2/3$.} This matches the tight bound $\Theta ( N^{1/3} )$ for $\ell =2$ \cite{Zha15} and gives a new upper bound $O(N^{10/21})$ for $\ell =3$, which is asymptotically lower than the trivial bound $O(N^{1/2})$ (see \autoref{sec:mult-grov}). The white circles for $\ell = 2,3$ in \autoref{fig:Known-Bounds-Rnd} correspond to this algorithm. 
In the case of $\ell \geq 4$, however, 
$(N^{1-2^{\ell-2}/(2^\ell-1)})^{(\ell - 1)/\ell}$ is asymptotically greater than the trivial bound $N^{1/2}$.

\subsection{Previous Works on Multiclaw}
As for quantum algorithms for finding multiclaws, there exist previous works on problems related to ours in the context of quantum computation~\cite{BHT97,DBLP:journals/tcs/Tani09,DBLP:conf/coco/BuhrmanDHHMSW01}, but those works usually focus on the worst-case complexity and regard the domain sizes of functions as the problem size.
To the best of authors' knowledge, there does not exist any previous work that studies $\ell$-claw-finding problem for general $\ell$ for random functions.
\section{Quantum Algorithm for Claw-Finding}
\label{sec:alg}
This section provides our quantum algorithm $\Mclaw$ that finds a multiclaw.
Since multiclaw-finding algorithms can be converted into multicollision-finding algorithms via a trivial reduction (see \autoref{lem:clawtocoll}), $\Mclaw$ can also be used to find multicollisions.

$\Mclaw$ finds an $\ell$-claw with $O(c_N N^{(2^{\ell-1}-1)/(2^\ell-1)})$ quantum queries for random functions $f_i \colon X_i \rightarrow Y$ for $i\in [\ell]$, where $\card{Y}=N$ and there exists a real value $c_N$ 
with $1\le c_N \in o(N^{{1}/({2^\ell -1})})$ such that ${N}/{c_N} \leq |X_i|$ holds for all $i\in [\ell]$.
This implies that, 
for a random function $f \colon X \rightarrow Y$ with $\card{Y}=N$ and $\card{X}\ge l\cdot N$,
an $\ell$-collision can be found with $O(N^{(2^{\ell-1}-1)/(2^\ell-1)})$ quantum queries.
Our bound is optimal for $\ell=2$ and improves upon the bound $O(N^{10/21})$ obtained in \autoref{sec:sok} by combining previous results.
In addition, our bound for the first time improves the simple bound of $O(N^{1/2})$ for all $\ell\ge 4$.

Our algorithm assumes without loss of generality that $\card{X_1},\dots,\card{X_\ell}$ are less than or equal to $\card{Y}$, since
it can also be applied to the functions of interest in the context of cryptography, i.e., the functions of which domains are much larger than their ranges, by restricting the domains of them to subsets. For instance, if $\card{X_i}$ is larger than $\card{Y}$ for some $i$, then one can pick a subset of $X'_i\subset X_i$ of size $\card{Y}$ for all such $i$ and apply the algorithm to 
the functions restricted to $X'_i$.
When we use the algorithm for finding a multicollision via the simple reduction, we can similarly assume $\card{X}\le \ell\cdot \card{Y}$:
If there exists an algorithm that finds an $\ell$-collision for a random function $F \sim U(\Func(X,Y))$ as long as $\card{X} = \ell \cdot \card{Y}$ , then we can use it to find an $\ell$-collision also in the case $\card{X} > \ell\cdot \card{Y}$ with the same number of queries and the same space by choosing a subset $X' \subset X$ of size $\ell\cdot \card{Y}$ and running the algorithm on $F|_{X'}$.

To simply show how the algorithm works,
we begin with 
the case of one-to-one functions.
Then, we describe the actual algorithm for random functions.
We need  elaborate analyses in optimizing parameters 
to deal with
random functions whose
domain sizes may be much smaller than their range size.

\subsection{Algorithm for One-to-One Functions}
The main idea of algorithm $\Mclaw$ is very simple: We just extend the strategy of the $\BHT$ algorithm.
Recall that the $\BHT$ algorithm~\cite{BHT97} first makes a list $L_1$ of many $1$-collisions, and then makes a 2-collision from $L_1$ with $\MTPS$.
Extending this strategy to $\ell$-collision finding, we first make a list $L_1$ of many $1$-collisions.
Then, for $i=2,\dots,\ell-1$ in this order, we make a list $L_i$ of many $i$-collisions from $L_{i-1}$ by iteratively applying $\BBHT$, and finally make an $\ell$-collision from $L_{\ell-1}$ by using $\BBHT$ again.
By optimizing the sizes of the lists, we obtain the query complexity $O(N^{(2^{\ell-1}-1)/(2^\ell-1)})$
if the domain sizes are linear in the range size.

To describe this more concretely together with complexity analysis,
we first explain the idea of the $\BHT$ algorithm adapted to the claw finding problem,
and then show how to develop a quantum algorithm to find $3$-claws from $\BHT$
 and how to extend it further to {the case of finding an $\ell$-claw for any $\ell$}.
For ease of understanding, we assume input functins are all \emph{one-to-one} throughout this subsection.

\subsubsection{Adaptation of the BHT algorithm.}
The BHT algorithm for the collision-finding problem described in \autoref{subsubsec:BHT_for_collision} can be adapted to the claw-finding problem
in a straight-forward manner as follows.

Let $f_1 \colon X_1 \rightarrow Y$ and $f_2 \colon X_2 \rightarrow Y$ be \tnote{one-to-one} functions.
The goal of the \BHT~algorithm is to find a ($2$-)claw for $f_1$ and $f_2$ with \tnote{$O(N^{1/3})$} quantum queries.
For simplicity, we assume that $\card{X_1} = \card{X_2} = \card{Y} = N$ holds.
Let $t_1$ be a parameter that defines the size of a list of $1$-claws for $f_1$. It will be set as $t_1=N^{1/3}$. 

First, collect $t_1$ many $1$-claws for $f_1$ and store them in a list $L_1$.
This first step makes $t_1$ queries.
Second, extend one of $1$-claws in $L_1$ to a $2$-claw for $f_1$ and $f_2$, by using $\BBHT$, and output the obtained $2$-claw.
This second step makes \tnote{$O(\sqrt{N/t_1})$} queries (see \autoref{thm:BBHT}).
Overall, the above algorithm makes $q_2(t_1) = t_1 + \sqrt{N/t_1}$ quantum queries \tnote{up to a constant factor}.
The function $q_2(t_1)$ takes its minimum value $2 \cdot N^{1/3}$ when $t_1 = N^{1/3}$.
By setting $t_1 = N^{1/3}$, the BHT algorithm is obtained.

\subsubsection{From \BHT~to a $3$-claw-finding algorithm.}
Next, we show how the above strategy of the $\BHT$ algorithm can be extended to develop a $3$-claw-finding algorithm.
Let $f_i \colon X_i \rightarrow Y$ be a {one-to-one} function for each $i\in \set{1, 2, 3}$.
Our goal here is to find a $3$-claw for $f_1$, $f_2$, and $f_3$ with $O(N^{3/7})$ quantum queries.
For simplicity, below we assume $|X_1| = |X_2| = |X_3|= |Y|=N$.
Let $t_1,t_2$ be parameters that define the number of $1$-claws for $f_1$ and that of $2$-claws for $f_1$ and $f_2$, respectively,
to collect in the algorithm (they will be fixed later).

First, collect $t_1$ many $1$-claws for $f_1$ and store them in a list $L_1$.
This first step makes $t_1$ queries.
Second, extend $1$-claws in $L_1$ to $t_2$ many $2$-claws for $f_1$ and $f_2$ by using \BBHT, and store them in a list $L_2$.
Here we do not discard the list $L_1$ until we construct the list $L_2$ of size $t_2$.
Since \BBHT~makes \tnote{$O(\sqrt{N/t_1})$} queries to make a $2$-claw from $L_1$, this second step makes $t_2 \cdot O(\sqrt{N/t_1})$ queries if $t_2=o(t_1)$ (see~\autoref{thm:BBHT}).
Finally, extend one of $2$-claws in $L_2$ to a $3$-claw for $f_1$, $f_2$, and $f_3$ by using \BBHT, and output the obtained $3$-claw.
This final step makes \tnote{$O(\sqrt{N/t_2})$} queries.
Overall, the above algorithm makes $q_3(t_1,t_2) = t_1 + t_2 \cdot \sqrt{N/t_1} +  \sqrt{N/t_2}$ quantum queries \tnote{up to a constant factor}.
The function $q_3(t_1,t_2)$ takes its minimum value $3 \cdot N^{3/7}$ when $t_1 = t_2 \cdot \sqrt{N/t_1} =  \sqrt{N/t_2}$, which is equivalent to $t_1 = N^{3/7}$ and $t_2 = N^{1/7}$.
By setting $t_1$ and $t_2$ to these values,
we can obtain a $3$-claw finding algorithm with {$O(N^{3/7})$} quantum queries.

\subsubsection{$\ell$-claw-finding algorithm for general $\ell$.}
By generalizing the above idea for finding a $3$-claw, we can find an $\ell$-claw for general $\ell$ as follows.
Let $f_i \colon X_i \rightarrow Y$ be a {one-to-one} function for each $i\in [\ell]$.
Our goal here is to find an $\ell$-claw for $f_1,\dots,f_\ell$.
For simplicity, we assume that $|X_1| = \cdots = |X_\ell|= |Y|=N$ holds.
Let $t_1,\dots,t_{\ell-1}$ be parameters with $t_i=o(t_{i-1})$ for $i=2,.\dots, l$.

First, collect $t_1$ many $1$-claws for $f_1$ and store them in a list $L_1$.
This first step makes $t_1$ queries.
In the $i$-th step for $i=2,\dots,\ell-1$, extend  $t_i$ many $(i-1)$-claws in \tnote{$L_{i-1}$} to $t_i$ many $i$-claws for $f_1,\dots,f_i$ by using \BBHT, and store them in a list $L_i$.
Here we do not discard the list $L_{i-1}$ until we construct the list $L_i$ of size $t_i$.
Since \BBHT~makes \tnote{$O(\sqrt{N/t_{i-1}})$} queries to make an $i$-claw from $L_{i-1}$, the $i$-th step makes \tnote{$t_i \cdot O(\sqrt{N/t_{i-1}})$} queries.
Finally, extend one of $(\ell-1)$-claws in $L_{\ell-1}$ to an $\ell$-claw for $f_1,\dots,f_\ell$ by using \BBHT, and output the obtained $\ell$-claw.
This final step makes \tnote{$O(\sqrt{N/t_{\ell-1}})$} queries.
Overall, this algorithm makes $q_\ell(t_1,\dots,t_{\ell-1}) = t_1 + t_2 \cdot \sqrt{N/t_1} + \cdots + t_{\ell-1}\cdot\sqrt{N / t_{\ell-2}} + \sqrt{N/t_{\ell-1}}$ quantum queries \tnote{up to a constant factor}.
The function $q_\ell(t_1,\dots,t_{\ell-1})$ takes its minimum value $\ell \cdot N^{(2^{\ell-1}-1)/(2^\ell-1)}$ when $t_1 = t_2 \cdot \sqrt{N/t_1} = \cdots = t_{\ell-1}\cdot\sqrt{N / t_{\ell-2}} = \sqrt{N/t_{\ell-1}}$, which is equivalent to $t_i = N^{(2^{\ell-i}-1)/(2^\ell-1)}$.
By setting $t_i$'s to these values,
we can find an $\ell$-claw with {$O(N^{(2^{\ell-1}-1)/(2^\ell-1)})$} quantum queries.

Our algorithm \Mclaw~is developed based on the above strategy so that 
it can deal with \emph{random} functions.

\subsection{Algorithm for Random Functions}
This subsection only describes how the algorithm works,
and a rigorous complexity analysis of \Mclaw~is provided in the next subsection.
Suppose that $|Y|=N$ for a sufficiently large positive integer $N$,
and that $\card{X_i} \leq \card{Y}$ for all $i$.

Our algorithm is parametrized by a positive integer $k\ge 2$, which controls the success probability and the complexity based on 
Markov's inequality.
We denote by $\Mclaw_k$ the algorithm for the parameter $k$.
$\Mclaw_k$ is applicable if there exists a positive real $c_N$
satisfying  $1\le c_N \in o(N^{{1}/({2^\ell - 1})})$ 
such that $|X_i|$ is at least $|Y| / c_N$ for each $i\in [\ell]$.
We impose an upper limit on the number of queries {that $\Mclaw_k$ is allowed to make}:
We design $\Mclaw_k$ in such a way that it immediately stops and aborts if the number of queries that have been made reaches the limit specified by the parameter $\mathsf{Qlimit}_k := k \cdot 169 \cdot \ell \cdot c_N \cdot  N^{({2^{\ell-1}-1})/({2^{\ell}-1})}$.
The upper limit $\mathsf{Qlimit}_k$ is necessary to prevent the algorithm from running forever, and to make the expected value of the number of queries converge.
We also define the parameters controlling the sizes of the lists:
\begin{equation*}
N_i \colon=
\begin{cases}
\frac{N}{4c_N} & (i=0), \\
N^{\frac{2^{\ell-i}-1}{2^\ell-1}} & (i \geq 1).
\end{cases}
\end{equation*}
For ease of notation, we define $L_0$ and $L_0'$ as $L_0 = L'_0 = Y$.
Then, $\Mclaw_k$ is described as in Algorithm~\ref{alg:Mclaw},
assuming that $\card{X_i}= \card{Y}/c_N $.
If there exist $i$'s such that $\card{X_i}> \card{Y}/c_N $,
we choose an arbitrary subset $X'_i\subset X_i$ of cardinality $\card{Y}/c_N $ for every such $i$
and then apply  Algorithm~\ref{alg:Mclaw} to the functions restricted to $X'_i$.
Since this preprocess requires no query and only negligible time/space
by choosing appropriate $X_i$, say, the first $\card{Y}/c_N$ elements of $X_i$,
we can safely assume $\card{X_i}= \card{Y}/c_N $ in the following analysis.

\begin{algorithm}
\caption{$\Mclaw_k$}
\label{alg:Mclaw}
\begin{algorithmic}
\REQUIRE $\set{f_i\sim U(\Func(X_i,Y))\colon \card{X_i} = \card{Y}/c_N,\  i\in [\ell]}$.
\ENSURE An $\ell$-claw for $f_1,\dots,f_\ell$ or $\perp$.
\renewcommand{\algorithmicensure}{\textbf{Stop condition:}}
\ENSURE If the number of queries reaches $\mathsf{Qlimit}_k$, stop and output $\perp$.
\STATE $L_0 = L'_0 = Y$, $L_1, \dots, L_{\ell} \gets \emptyset$, $L'_1, \dots, L'_{\ell} \gets \emptyset$.
\FOR{$i=1$ to $\ell$}
    \FOR{$j=1$ to $\left\lceil 4c_N \cdot N_i \right\rceil$}
            \STATE Find $x_j \in X_i$ such that $\{(x^{(1)}, \dots, x^{(i-1)}; f_i(x_j))\}\in L'_{i-1}$ 
            for some $(x^{(1)}, \dots, x^{(i-1)})\in X_1\times\dots \times X_{i-1}$
            by running $\MTPS$ on $f_i$ with $L'_{i-1}$, and let $y \colon= f_i(x_j)$.\quad \hfill 
            
        \STATE $L_i \gets L_i \cup \{(x^{(1)}, \dots, x^{(i-1)},x_j; y)\}$, $L'_i \gets L'_i \cup \{y\}$.
        \STATE $L_{i-1} \gets L_{i-1} \setminus \{(x^{(1)}, \dots, x^{(i-1)}; y)\}$, $L'_{i-1} \gets L'_{i-1} \setminus \{y\}$.
    \ENDFOR
\ENDFOR
\STATE Return an element $(x^{(1)},\dots,x^{(\ell)};y) \in L_\ell$ as an output.
\end{algorithmic}
\end{algorithm}

\subsection{Complexity Analysis for Random Functions}

The goal of this section is to show the main theorem,
from which \autoref{thm:claw_informal} follows via the reduction provided in \autoref{lem:clawtocoll}.

\begin{theorem}[\bf main]\label{thm:main}
Let $N$ be a sufficiently large positive integer, and let $c_N$
be any fixed real satisfying $1\le c_N\in o(N^{{1}/({2^\ell-1})})$.
Then,
for $\ell$ functions 
$\set{f_i\sim U(\Func(X_i,Y))\colon i\in [\ell]}$,
where $\card{Y}=N$ and $\card{X_i}\ge N/c_N$,
$\Mclaw_k$ finds an $\ell$-claw with probability at least 
$1-\varepsilon(\ell, N, c_N)-1/k$, where
\begin{equation}
\varepsilon(\ell, N, c_N):= \frac{2\ell}{N}  + \ell \cdot \exp\left(- \frac{1}{25} \cdot \frac{N^{\frac{1}{2^\ell - 1}}}{c_N}\right)
\label{eq:MclawkProbLower}
\end{equation}
and the probability is taken over both the inherent randomness of $\Mclaw_k$
and the randomness of choices of the $\ell$ functions $f_i$.
Moreover, $\Mclaw_k$ makes at most
\begin{equation*}
\mathsf{Qlimit}_k := k \cdot  169 \cdot\ell \cdot c_N \cdot  N^{\frac{2^{\ell-1}-1}{2^{\ell}-1}} 
\end{equation*}
quantum queries,
and runs
in $\tilde{O}(\mathsf{Qlimit}_k)$ time 
on $\tilde{O}(\mathsf{Qlimit}_k)$ qubits
for every possible $\ell$ functions,
where $\tilde{O}(\cdot)$ suppresses a $\log N$ factor.
\end{theorem}
This theorem shows that, for each integer $k \geq 2$, $\Mclaw_k$ finds an $\ell$-claw with a constant probability by making $O\big(c_N\cdot N^{({2^{\ell-1}-1})/({2^{\ell}-1}})\big)$ queries.

Before proving the theorem, we first show the following two lemma for later use.
\begin{lemma}\label{lem:BallBin}
Let $X,Y$ be non-empty finite sets such that $\card{X} \leq \card{Y}$.
Suppose that a function $f \colon X \rightarrow Y$ is chosen uniformly at random
from $\Func(X,Y)$, that is, $f\sim U(\Func(X,Y))$.
Then, it holds that
\begin{equation*}
\Pr_{f \sim U(\Func(X,Y))} \left[ \card{{\rm Im}(f)} \geq  \frac{|X|}{2} - \sqrt{|X|\ln |Y|/2} \right] \geq 1 - \frac{2}{|Y|}.
\end{equation*}
\end{lemma}
\begin{proof}
Note that, for each $x \in X$, $f(x)$ is the random variable that takes value in $Y$.
Moreover, $\{f(x)\}_{x \in X}$ is the set of independent random variables.
Let us define a function $\Phi \colon Y^{\times \card{X}} \rightarrow \Natural$ by $\Phi \left(y_1,\dots,y_{\card{X}}\right) = \left| Y \setminus \{y_i\}_{1 \leq i \leq \card{X}}\right|$.
Then $\Phi$ is $1$-Lipschitz, i.e., 
\begin{equation*}
\left| \Phi(y_1,\dots,y_{i-1},y_i,y_{i+1},\dots,y_{\card{X}}) - \Phi(y_1,\dots,y_{i-1},y'_i,y_{i+1},\dots,y_{\card{X}}) \right| \leq 1
\end{equation*}
holds for arbitrary choices of elements $y_1,\dots,y_{\card{X}}$, and $y'_i$ in $Y$.
Now we apply \autoref{th:McDiarmid} to $\Phi$
with
$M = \card{X}$, $\lambda = \sqrt{|X|\ln|Y| / 2}$, and $y_x := f(x)$ for each $x \in X$ (here we identify $X$ with the set $\{1,\dots,\card{X}\}$).
Then, since it holds that $\E \left[ \Phi(y_1,\dots,y_M) \right] = |Y|\left( 1- 1/\card{Y}\right)^{\card{X}}$, we have that
\begin{align*}
\Pr_{f \sim U(\Func(X,Y))}\left[ \Phi(y_1,\dots,y_M) \geq \card{Y}\left( 1- 1/\card{Y}\right)^{\card{X}} + \sqrt{|X|\ln|Y| / 2} \right] \leq \frac{2}{\card{Y}}.
\end{align*}
In addition, it follows that
\begin{align*}
\card{Y}\left( 1- 1/\card{Y}\right)^{\card{X}} &\leq \card{Y}e^{-\card{X}/\card{Y}} \leq \card{Y} \left(1 - \frac{|X|}{|Y|} + \frac{1}{2}\left(\frac{|X|}{|Y|}\right)^2 \right) \nonumber \\
&= |Y| - |X| \left( 1 - \frac{1}{2}\frac{|X|}{|Y|}\right) \leq |Y| - \frac{\card{X}}{2},
\end{align*}
where we used the assumption that $\card{X} \leq \card{Y}$ for the last inequality.
Since it holds that $\Phi(y_1,\dots,y_M) = \left|Y \setminus {\rm Im}(f) \right|$ and $\left|{\rm Im}(f)\right| = |Y| -  \left|Y \setminus {\rm Im}(f) \right|$, it follows that 
$\left|{\rm Im}(f)\right|$ 
is at least
\begin{align*}
 |Y|  - \left(\card{Y} - \frac{\card{X}}{2} + \sqrt{\card{X} \ln \card{Y} / 2} \right) = \frac{\card{X}}{2} -  \sqrt{\card{X} \ln \card{Y} / 2}
\end{align*}
with a probability at least $1 - 2/{\card{Y}}$, which completes the proof.
\end{proof}
\begin{lemma}\label{lm:iteration cost}
Let $S=(\upp{S}{1},\dots, \upp{S}{k})\in (\set{0,1}^{n})^k$ be a sorted list of $k$ $n$-bit strings.
For the list $S$
stored in a quantum register $\reg{S}$
and an oracle $\calO_f$ 
for $f\colon \set{0,1}^m\to \set{0,1}^n$
such that $\calO_f \ket{x}\ket{b}=\ket{x}\ket{b\oplus f(x)}$ for every $x\in \set{0,1}^m$ and $b\in \set{0,1}^n$,  a single Grover iteration $G_f=-WS_0WS_f$ (following the notation in \cite{boyer1998tight}) of Grover's search algorithm that finds $x$ such that $f(x)\in S$
can be performed in $O(m+n\log k)$ time with $m$ qubits for indexing $x$ and $O(m+kn)$ ancilla qubits (including those for $\reg{S}$),
where each of the ancilla qubits is assumed to be in the initial state $\ket{0}$ and is returned to the original state $\ket{0}$ after the iteration.
\end{lemma}
\begin{proof}
Grover's search algorithm starts with the state $\sum_{x=0}^{2^m-1}\frac{1}{\sqrt{2^m}}\ket{x}\ket{0^n}$ on an $(m+n)$-qubit register. 
Let us call the first $m$ qubits the \emph{index register}.
The algorithm requires additional qubits to perform $G_f$.
In the following, we focus on this additional qubits when considering the number of required qubits.

The operator $W$ can be performed in $O(m)$ time, since $W$ applies a one-qubit
Hadamard gate to each qubit in the index register. The operator $S_0$ adds the phase $-1$ to the all-zero state in the index register, which can be performed in $O(m)$ time
with $O(m)$ ancilla qubits.  
To perform $S_f$, 
we first make a query to $\calO_f$ to tranform $\ket{x}\ket{0^n}\mapsto \ket{x}\ket{f(x)}$,
and then perform the following transformation 
\[
 \ket{x}\ket{f(x)}\ket{S}\stackrel{(1)}{\longmapsto} (-1)^{[f(x)\in S]}\ket{x}\ket{f(x)}\ket{S}
\stackrel{(2)}{\longmapsto} (-1)^{[f(x)\in S]}\ket{x}\ket{0^n}\ket{S}.
\]
The transformation $(1)$ takes $O(n\log k)$ time with $O(kn)$ qubits,
since 
whether $f(x)$ in $S$
can be checked by
performing binary search coherently
on the sorted list $S$.
The transformation $(2)$ is done by making a query again.
Therefore, the total time is $O(m+n\log k)$, and the total number of qubits is $O(m+ kn)$.
\end{proof}

Now, we prove~\autoref{thm:main}.

\begin{proofof}{\autoref{thm:main}}
We show first that~\autoref{eq:MclawkProbLower} holds.
Let us define $\good^{(i)}$ to be the event that
\begin{equation*}
\card{\Img(f_i) \cap L_{i-1}'} \geq N_{i-1}
\end{equation*}
holds just before $\Mclaw_k$ starts to construct $i$-claws.
We prove the following claim later.
\begin{claim}
\label{cl:PrGood(i)}
For sufficiently large $N$, it holds that
$\Pr\big[\good^{(i)}\big] \geq 1 - \frac{2}{N}  - \exp\big(- \frac{1}{25} \cdot \frac{N_{i-1}}{c_N}\big).$
\end{claim}
\noindent Let $\good$ denote the event $\good^{(1)} \land \cdots \land \good^{(\ell)}$.
Notice that $\Pr\left[\neg\good\right]$ is upper-bounded by $\sum_{i=1}^\ell \Pr[\neg \good^{(i)}]$.
Thus, it follows from
\autoref{cl:PrGood(i)}
that
\begin{align}
\Pr\left[\neg\good\right]\le \sum_{i=1}^\ell \left( \frac{2}{N}  + \exp\left(- \frac{1}{25} \cdot \frac{N_{i-1}}{c_N}\right) \right)
\leq \ell \left[\frac{2}{N}  + \cdot \exp\left(- \frac{1}{25} \cdot \frac{N^{\frac{1}{2^\ell - 1}}}{c_N}\right)\right]:=
\varepsilon(\ell,N,c_N). \label{eq:GoodProbBound}
\end{align}

We prove the following claim later, which shows
that the average query complexity over the inherent randomness of the algorithm
is at most $ \frac{1}{k}\mathsf{Qlimit}_k$ when the event $\good$ occurs.

\begin{claim}
\label{cl:E[Q|good]}
For sufficiently large $N$, it holds that
$\E\left[ Q \mymiddle \good \right] \leq \frac{1}{k}\mathsf{Qlimit}_k$,
where $Q$ is the total number of queries made by $\Mclaw_k$.
\end{claim}
\noindent It follows from \autoref{cl:E[Q|good]} that $\E[Q]$ is upper-bounded by
\begin{align*}
\E[Q \mid  \good] + \E[Q \mid \lnot \good] \Pr[\lnot \good] 
&\leq \left( \frac{1}{k} + \Pr[\lnot \good]  \right) \cdot \mathsf{Qlimit}_k \label{eq:ExpectQupper}.
\end{align*}
From Markov's inequality, the probability that $Q$ reaches $\mathsf{Qlimit}_k$ 
is at most
\begin{equation}
\Pr \left[ Q \geq \mathsf{Qlimit}_k \right] \leq \frac{\E[Q]}{\mathsf{Qlimit}_k} \leq \frac{1}{k} + \Pr[\lnot \good]\leq  \frac{1}{k} +   \varepsilon(\ell, N, c_N), \label{eq:ProbReachLimit}
\end{equation}
where the last inequality follows from~\autoref{eq:GoodProbBound}.
The event ``$Q$ does not reach $\mathsf{Qlimit}_k$'' implies that $\Mclaw_k$ finds an $\ell$-claw.
Thus, 
from~\autoref{eq:ProbReachLimit}, 
$\Mclaw_k$ finds an $\ell$-claw with probability at least $1 - \varepsilon(\ell, N, c_N)-{1}/{k}$
by making at most $\mathsf{Qlimit}_k$ queries
for every input.

Next, we move on the time and space complexity.
To make $\Mclaw_k$ time-efficient,
we keep $L_i$ and $L'_i$ sorted with respect to $y$ values
by using an appropriate data structure such as balanced trees.
Since it takes only a polylogarithmic time in $N$ for each insertion or deletion,
the time complexity for updating $L_i$ and $L'_i$ is negligible.
We also assume that $L_i$ and $L'_i$ are stored in quantum registers
until  $L_{i+1}$ and $L'_{i+1}$ are constructed (and then they are discarded).
Thus, the total number of qubits for storing $L_{i}$ and $L'_{i}$ for $i\in [n]$
is the maximum of $\card{L_{i}}$ over $i$ up to a logarithmic factor.

The dominant part is thus the multi-target quantum search,
which consists of Grover iterations.
Hence, we estimate the time and space (qubits) required by each Grover iteration 
with  Lemma~\ref{lm:iteration cost}.
For this, set the parameters $(m,n,k,S)$ in the lemma to $(\ceil{\log \card{X_i}}, \ceil{\log \card{Y}}, \card{L_{i-1}}, L_{i-1})$.
Then, each Grover iteration in the $j$-th search during the construction of  $i$-claws
takes $O(\log \card{X_i}+ \log\card{Y}\cdot \log\card{L_{i-1}})=\tilde{O}(1)$ time,
where we use $\card{X_i}\le \card{Y}$.
Thus, the total time is upper-bounded by the number of Grover iterations,
which is further upper-bounded by the number of queries, 
where we ignore logarithmic factors. 

For the space complexity, 
each Grover iteration in the $j$-th search during the construction of  $i$-claws
uses
$O(\log |X_i|+(c_NN_{i-1}-j+1) \log |Y|)=\tilde{O}(c_NN_{i-1}-j+1)$ qubits.
Since all qubits used in a single iteration can be reused in the next iteration,
the number of  qubits required in constructing $i$-claws is
$\tilde{O}(\max_{j}(c_NN_{i-1}-j+1))=\tilde{O}(c_NN_{i-1})$.
Thus, the number of qubits required in $\Mclaw_k$
is $\tilde{O}(\max_{i\in [\ell]}c_N\cdot N_{i-1})=\tilde{O}(c_N\cdot N_0)=\tilde{O}(N)$. 
However, this is too large.  The reason is that $k=c_NN_{i-1}-j+1$ is large when $i=1$.

We can substantially reduce the number of the qubits required for collecting $1$-claws
by using the simple fact that searching for $x$ with $f_1(x)\in L'_{0}$ is 
equivalent to searching for $x$ with $f_1(x)\not\in Y\setminus L'_{0}$.
More concretely,  instead of flipping the phase of $\ket{x}$ such that $f_1(x)\in L'_0$,
we flip the phase $\ket{x}$ such that $f_1(x)\in Y\setminus L'_0$ and then flip the phase of all basis states:
\[
\ket{x}\ket{f(x)}\ket{Y\setminus L'_{0}}\mapsto  (-1) (-1)^{[f(x)\in Y\setminus L'_{0}]}\ket{x}\ket{f(x)}\ket{Y\setminus L'_{0}}=(-1)^{[f_1(x)\not\in Y\setminus L'_{0}]}\ket{x}\ket{f(x)}\ket{Y\setminus L'_{0}}.
\]
The point is that $|Y\setminus L'_{0}|$ (corresponding to $k$ in Lemma~\ref{lm:iteration cost}) is small.
The cardinality $|Y\setminus L'_{0}|$ in the $j$-th search in collecting $1$-claws
is equal to $j-1$ for $j=1,\dots, \lceil 4c_NN_1\rceil$.
Hence, the number of qubits required for collecting $1$-claws
is $\tilde{O}(c_N N_1)$.
Therefore, the total number of required qubits is 
$
\tilde{O}\big(\max\big\{c_N \cdot N_1,\max_{i\ge 2}c_N\cdot N_{i-1}\big\}\big)=\tilde{O}\left(c_N \cdot N_1\right)=\tilde{O}\big(c_N\cdot N^{({2^{l-1}-1})/({2^{l}-1})}\big).$
\end{proofof}

\vspace{12pt}
Finally, we provide the proofs of the two claims.

\begin{proofof}{\autoref{cl:PrGood(i)}}
Suppose that $\Mclaw_k$ has finished making $L_{i-1}$ but has not started to make $i$-claws yet.
Hence, it holds that $\card{L_{i-1}}=\card{L'_{i-1}}=\left\lceil4c_N N_{i-1}\right\rceil$.
Let $\pregood^{(i)}$ be the event that $\card{{\rm Im}(f_i)} \geq \left\lceil N/3c_N \right\rceil$ holds.
It follows from $c_N\ge 1$
that ${\card{X_i}}/{2} - \sqrt{\card{X_i}\ln\card{Y}/2} \geq \lceil {N}/({3c_N}) \rceil$ holds for sufficiently large $N$.
This together with \autoref{lem:BallBin} implies that
\begin{equation*}
\Pr\left[ \pregood^{(i)} \right] \geq 1 - \frac{2}{\card{Y}}.
\end{equation*}
Let us identify $X_i$ and $Y$ with the sets $\{1,\dots,\card{X_i}\}$ and $\{1,\dots,\card{Y}\}$, respectively.
Let $z_j$ be the $j$-th element in ${\rm Im}(f_i)$.
Let $\chi_j$ be the indicator variable that is defined as $\chi_j = 1$ if and only if $z_j \in L'_{i-1}$, and define a random variable $\chi$ as $\chi:= \sum_{j=1}^{\card{{\rm Im}(f_i)}} \chi_j$.
Then, $\chi$ follows the hypergeometric distribution.
We thus apply \autoref{th:BoundOnHyperGeometricDistribution} with $n_1 = \left\lceil N / 3c_N\right\rceil$, $n=N$, and $m = |L'_{i-1}| = \left\lceil 4c_N N_{i-1} \right\rceil$ for the random variable $\chi$ under the condition that $\card{{\rm Im}(f_i)} = \left\lceil N/3c_N\right\rceil$.
Let $\equal$ denote the event that $\card{{\rm Im}(f_i)} = \left\lceil N/3c_N \right\rceil$ holds.
Then, we have 
\begin{align*}
&\Pr\left[ \chi - \E\left[\chi \middle| \equal \right] < -\frac{1}{4}  \E\left[\chi \middle| \equal\right] \middle| \equal \right] \nonumber \\
&\quad \leq
\exp\left( -2\left( \frac{1}{m+1} + \frac{1}{n-m+1}\right) \left( \right( \E\left[\chi\middle|\equal\right] / 4 \left)^2 - 1 \right)\right) \nonumber \\
&\quad \leq \exp\left( -\frac{1}{10}\frac{1}{m} \left( \E\left[\chi\middle|\equal\right]\right)^2 \right) \leq \exp\left(- \frac{1}{25} \cdot \frac{N_{i-1}}{c_N}\right).
\end{align*}
This implies $\Pr\left[ \chi \geq N_{i-1} \middle| \equal \right] \geq 1 - \exp\big(- {N_{i-1}}/(25{c_N})\big)$
by using $\E\left[\chi\middle|\equal\right]= \frac{n_1m}{n} \geq \frac{4}{3}{N_{i-1}}$.
Hence, we have
\begin{align*}
\Pr\left[\left| {\rm Im}(f_i) \cap L'_{i-1}\right| \geq N_{i-1} \middle| \pregood^{(i)} \right]
&=
\Pr\left[\chi \geq N_{i-1} \middle| \pregood^{(i)} \right] \nonumber \\
&\geq
\Pr\left[\chi \geq N_{i-1} \middle| \equal \right] \nonumber \\
&\geq 1 - \exp\left(- \frac{1}{25} \cdot \frac{N_{i-1}}{c_N}\right).
\end{align*}
Therefore, it follows that
\begin{align*}
\Pr\left[ \good^{(i)} \right]
&>
\Pr\left[ \good^{(i)} \middle| \pregood^{(i)} \right] \cdot \Pr\left[ \pregood^{(i)} \right] \nonumber\\
&= \Pr\left[ \left| {\rm Im}(f_i) \cap L'_{i-1}\right| \geq N_{i-1}  \middle| \pregood^{(i)} \right] \cdot \Pr\left[ \pregood^{(i)} \right] \nonumber\\
&\geq \left(1 - \frac{2}{|Y|} \right)\left( 1 - \exp\left(- \frac{1}{25} \cdot \frac{N_{i-1}}{c_N}\right) \right) \nonumber \\
&\geq 1 - \frac{2}{|Y|}  - \exp\left(- \frac{1}{25} \cdot \frac{N_{i-1}}{c_N}\right).
\end{align*}
\end{proofof}

\vspace{12pt}
\begin{proofof}{\autoref{cl:E[Q|good]}}
Let us fix $i$ and $j$.
Let $Q^{(i)}_j$ denote the number of queries made by $\Mclaw_k$ in performing the $j$-th search 
to construct
the list of $i$-claws, and let $Q^{(i)}:=\sum_{j=1}^{\ceil{4c_N\cdot N_i}} Q^{(i)}_j$.
In the $j$-th search to construct the list of $i$-claws, we search $X_i$ for $x$ with $f_i(x)\in L'_{i-1}$, where there exist at least $|L'_{i-1} \cap {\rm Im}(f_i)| \geq N_{i-1} -j+1$ answers in $X_i$ under the condition that $\good^{(i)}$ occurs.
From~\autoref{cor:MTPS}, the expected number of queries made by $\MTPS$ in the $j$-th search
is upper-bounded by
\begin{align*}
9 \sqrt{5\card{X_i}/\card{f^{-1}_i(L'_{i-1})}} &\leq 9 \sqrt{5\card{X_i}/\card{L'_{i-1} \cap {\rm Im}(f_i)}} 
\leq 21 \sqrt{N/(c_N\cdot N_{i-1})}
\end{align*}
for each $j$ under the condition that $\good^{(i)}$ occurs, 
where we used the condition that $N_{i-1} = \omega(c_N N_i)$ holds 
and the assumption that $\card{X_i}=N/c_N$
for the last inequality.
Hence, it follows that
\begin{align*}
 \E\left[Q^{(i)} \mymiddle \good^{(i)}\right]
 &= \E\left[\sum_{j=1}^{ \lceil 4c_N N_i \rceil} Q_j^{(i)} \mymiddle \good^{(i)}\right] = \sum_{j=1}^{ \lceil 4c_N N_i \rceil} \E\left[Q_j^{(i)} \mymiddle \good^{(i)}\right] \\
 &\leq \sum_{j=1}^{ \lceil 4c_N N_i \rceil} 21 \sqrt{{N}/(c_N\cdot N_{i-1})} 
 \leq
 \begin{cases}
 169\cdot c_N\cdot N^{\frac{2^{\ell-1}-1}{2^\ell-1}} & (i=1)\\
 85 \cdot c_N^{1/2}\cdot N^{\frac{2^{\ell-1}-1}{2^\ell-1}} & (i\geq 2).
 \end{cases}
\end{align*}
This implies that $\E[Q \mid  \good]
 = \sum_{i} \E\left[Q^{(i)} \mymiddle \good^{(i)}\right]$ is upper-bounded by 
\begin{align*}
169\cdot c_N\cdot N^{\frac{2^{\ell-1}-1}{2^{\ell}-1}} 
     + \sum_{i=2}^{\ell} 85\cdot c_N^{1/2}\cdot N^{\frac{2^{\ell-1}-1}{2^{\ell}-1}}
 &\leq 169 \cdot \ell \cdot c_N \cdot  N^{\frac{2^{\ell-1}-1}{2^{\ell}-1}}  = \frac{1}{k} \mathsf{Qlimit}_k,
\end{align*}
which completes the proof.
\end{proofof}


\section{Discussions}\label{sec:discussions}

\subsection{Intermediate Measurements}\label{sec:IntMeas}
In describing the algorithm,
we assumed, for ease of analysis and understanding, that intermediate measurements were allowed.
For some implementations, it might be better to move all measurements to the end of the algorithm.
This is possible by the standard techniques as sketched in the following: All classical deterministic operations can be performed by quantum gates: Toffoli gates and $X$ gates, which act as AND gates and NOT gates, respectively (here we assume that available ancillary qubits are only the ones initialized to $\ket{0}$).
The only thing we need to concern is how to randomly determine the number of Grover iterations to be performed in each multi-target quantum search (\autoref{cor:MTPS}). Note that, in general, we can think of any randomized algorithm as deterministic one once the random bit string used in it is fixed.  For a random bit string $r$ of polynomial length, let $Q(r)$ be the $\ell$-claw (or $\ell$-collision) finding algorithm in which $r$ is used as a random bit string in determining the number of Grover iterations. For any fixed $r$,  $Q(r)$ consists of deterministic parts and quantum parts. Since the deterministic part can be quantized as stated above, $Q(r)$ can be represented as a single quantum circuit. Thus, we can modify the algorithm as follows: first prepare a uniform superposition over all possible $r$ in a quantum registers $\mathsf{R}$ (by applying Hadamard gates on the qubits in $\mathsf{R}$, which are initialized to $\ket{0}$),
and the all-zero state in a quantum register $\mathsf{W}$,
then run $Q(r)$ on $\mathsf{W}$ 
if the content of $\mathsf{R}$ is $r$ (more precisely,
run $Q(r)$ controlled by the content $r$ of $\mathsf{R}$),
and finally measure $\mathsf{W}$. From the above discussion, the success probability is exactly the same as that of the original algorithm with intermediate measurements. 
With this modification, the required space increases 
by the size of register $\mathsf{R}$ plus the number of ancillary qubits
required to quantize the deterministic part
and to make each gate in $Q(r)$ controlled by the content $r$ of $\mathsf{R}$.

\subsection{A Recursive Multicollision-Finding Algorithm}
\label{subsec:A_Recursive_Multicollision-Finding_Algorithm}
This section provides a recursive algorithm that finds an $\ell$-collision of a random function $f:X\rightarrow Y$ such that $|X| \geq \ell \cdot |Y|=\ell\cdot N$ by making $O\big( N^{(3^{\ell-1}-1)/(2\cdot 3^{\ell-1})} \big)$ queries and using $\tilde{O}(N^{1/3})$ qubits in time $\tilde{O}\big( N^{(3^{\ell-1}-1)/(2\cdot 3^{\ell-1})} \big)$ for any positive integer constant $\ell$.%
\footnote{The recursive algorithm presented here is from 
this paper's first conference version~\cite{DBLP:conf/asiacrypt/HosoyamadaSX17}.}
Let us denote this algorithm by 
$\rcoll$. 
The quantum query complexity and time complexity of $\rcoll$ are worse than those of $\Mclaw$ 
for $\ell \geq 3$.
However, $\rcoll$ uses only $\tilde{O}(N^{1/3})$ qubits for any $\ell$, whereas $\Mclaw$
 requires $\tilde{O}(N^{(2^{\ell-1}-1)/(2^l -1)})$ qubits (which is worse than $\tilde{O}(N^{1/3})$ for all $\ell \geq 3$) to find an $\ell$-collision.
Below we give only a rough idea for the algorithm $\rcoll$, and omit the detailed analysis for complexity and success probability. 

\paragraph{Idea of the recursive algorithm $\rcoll$.}
Let $\rcoll(i)$ be the algorithm $\rcoll$ that finds an $i$-collision for a random function $F\sim U(\Func(X,Y))$.
The algorithm $\rcoll(1)$  takes $x \in X$ randomly, make the query $x$ to the oracle $F$, and returns $(x,F(x))$.
For $i \geq 2$, 
$\rcoll(i)$ runs the following two procedures:
\begin{enumerate}
\item recursively call the algorithm $\rcoll(i-1)$ $t_{i-1}$ times to find $t_{i-1}$ many $(i-1)$-collisions, and make a list $L_{i-1}$ of $t_{i-1}$ many $(i-1)$-collisions (parameter $t_{i-1}$ will be fixed later).
\item perform the multi-target quantum search to extend one of $(i-1)$-collisions in $L_{i-1}$ to an $i$-collision.
\end{enumerate}
The first step makes $t_{i-1} \cdot q_{i-1}$ quantum queries, and the second step makes roughly $\sqrt{N/t_{i-1}}$ quantum queries, where $q_{i-1}$  denotes the number of quantum queries made by $\rcoll(i-1)$.
Thus, $\rcoll(i)$ makes roughly $q_i = t_{i-1} \cdot q_{i-1} + \sqrt{N/t_{i-1}}$ quantum queries in total.

To optimize the quantum query complexity, we set $t_{i}$ so that $t_{i}\cdot q_{i} = \sqrt{N/t_{i}}$ holds, which is equivalent to $t_i = (N/q^2_i)^{1/3}$ for all $i$.
Straight forward calculations show that the number of queries $q_l$ required to find an $\ell$-collision is optimized to be $N^{(3^{\ell-1}-1)/(2\cdot 3^{\ell-1})}$ by setting $t_i = N^{1/3^i}$ for each $i\in [\ell-1]$ (here we ignore constant multiplicative factors).

In the same way as $\Mclaw_k$, we can show that $\rcoll(\ell)$ can be 
performed in $\tilde{O}(N^{(3^{\ell-1}-1)/(2\cdot 3^{\ell-1})} )$ expected time
on $\tilde{O}(N^{1/3})$ qubits
 (note that $\tilde{O}(N^{1/3})$ suffices since $\max_{1 \leq i \leq \ell}|L_i|$ is in $O(N^{1/3})$).


\section{Conclusion}
\label{sec:conclusion}
Finding multicollisions is one of the most important problems in cryptology, both for attack and provable security. In the post-quantum era, this problem needs to be studied in a quantum setting to realize quantum-secure cryptographic schemes. We systematized knowledge on the multicollision-finding problem in a quantum setting and proposed a new quantum multicollision-finding algorithm for random functions.
For any $1\le c_N \in o(N^{{1}/({2^\ell-1})})$, our algorithm finds an $\ell$-collision of a random function $F \colon [M] \rightarrow [N]$ with $O\big(c_N\cdot N^{(2^{\ell-1}-1)/(2^\ell-1)}\big)$ quantum queries on average and it runs
$\tilde{O}\big(c_N\cdot N^{(2^{\ell-1}-1)/(2^\ell-1)}\big)$ expected time on
$\tilde{O}\big(c_N\cdot N^{(2^{\ell-1}-1)/(2^\ell-1)}\big)$ qubits, 
where $M$ is at least $\ell\cdot N/c_N$ and $\ell$ is a constant.
In particular, the complexities are $O\big(N^{(2^{\ell-1}-1)/(2^\ell-1)}\big)$ and $\tilde{O}\big( N^{(2^{\ell-1}-1)/(2^\ell-1)}\big)$, respectively, if $c_N$ is a constant.
The quantum query complexity matches the known tight bound for $\ell=2$, improves the simple combination of Zhandry and Belovs' results for $\ell=3$, and for the first time improves the trivial bound of $O(N^{1/2})$ for $\ell\ge 4$.

Actually, we provide a quantum algorithm that find an $\ell$-claw of $\ell$ random functions  ${f_i \colon [M_i] \rightarrow [N]}$ for $1 \leq i \leq \ell$ with the same expected quantum query complexity $O\big(c_N\cdot N^{(2^{\ell-1}-1)/(2^\ell-1)}\big)$
and the same expected time/space complexity $\tilde{O}\big(c_N\cdot N^{(2^{\ell-1}-1)/(2^\ell-1)}\big)$, where $M_i$ is at least $N/c_N$.
The multicollision-finding algorithm is obtained 
from this multiclaw-finding algorithm via a simple reduction.

There are still some open problems. 
The parameter $c_N$ controls the ratio of the domain size against the range size,
namely, as $c_N$ gets larger, the domain becomes smaller.
It is known that there is at least one $\ell$-collision with high probability if $c_N\le d_\ell \cdot N^{1/\ell}$ for some constant $d_\ell$. However, our algorithm works only for $c_N\in o(N^{{1}/({2^\ell-1})})$.
Thus, it would be interesting to seek an improved algorithm for all $c_N\le d_\ell \cdot N^{1/\ell}$. 
Second, the lower bound 
$\Omega (N^{(2^{\ell-1}-1)/(2^\ell-1)})$ by Liu and Zhandry~\cite{DBLP:journals/iacr/LiuZ18}
is optimal if $c_N$ is constant. Then, is it possible to improve the lower bound 
in the case of $c_N=\omega(1)$?
Third, we showed that the space complexity of our algorithm is the same as the query complexity up to a logarithmic factor and discussed an idea to reduce the space complexity at the sacrifice of query complexity.
It would be interesting to investigate more throughly the trade-offs between space and query/time complexities.

\bibliographystyle{alpha}
\bibliography{ms}

\newcommand{\etalchar}[1]{$^{#1}$}
\begin{thebibliography}{KMRT09}

\bibitem[Amb05]{Amb05}
Andris Ambainis.
\newblock Polynomial degree and lower bounds in quantum complexity: Collision
  and element distinctness with small range.
\newblock {\em Theory of Computing}, 1:37--46, 2005.
\newblock See \url{https://arxiv.org/abs/quant-ph/0304162}.

\bibitem[Amb07]{Amb07}
Andris Ambainis.
\newblock Quantum walk algorithm for element distinctness.
\newblock {\em {SIAM} J. Comput.}, 37(1):210--239, 2007.
\newblock The preliminary version appeared in FOCS 2004. See
  \url{https://arxiv.org/abs/quant-ph/0311001}.

\bibitem[AS04]{AS04}
Scott Aaronson and Yaoyun Shi.
\newblock Quantum lower bounds for the collision and the element distinctness
  problems.
\newblock {\em Journal of the ACM}, 51(4):595--605, 2004.
\newblock See \url{https://arxiv.org/abs/quant-ph/0112086}.

\bibitem[BBHT98]{boyer1998tight}
Michel Boyer, Gilles Brassard, Peter H{\o}yer, and Alain Tapp.
\newblock Tight bounds on quantum searching.
\newblock {\em Fortschritte der Physik: Progress of Physics}, 46(4-5):493--505,
  1998.
\newblock See \url{https://arxiv.org/abs/quant-ph/9605034}.

\bibitem[BCJ{\etalchar{+}}13]{BCJKM13}
Aleksandrs Belovs, Andrew~M. Childs, Stacey Jeffery, Robin Kothari, and
  Fr{\'{e}}d{\'{e}}ric Magniez.
\newblock Time-efficient quantum walks for $3$-distinctness.
\newblock In {\em Proceedings of the 40th International Colloquium on Automata,
  Languages, and Programming, {ICALP} 2013, Part {I}}, pages 105--122, 2013.
\newblock See \url{http://arxiv.org/abs/1302.3143} and
  \url{http://arxiv.org/abs/1302.7316}.

\bibitem[BDH{\etalchar{+}}01]{DBLP:conf/coco/BuhrmanDHHMSW01}
Harry Buhrman, Christoph D{\"{u}}rr, Mark Heiligman, Peter H{\o}yer,
  Fr{\'{e}}d{\'{e}}ric Magniez, Miklos Santha, and Ronald de~Wolf.
\newblock Quantum algorithms for element distinctness.
\newblock In {\em Proceedings of the 16th Annual {IEEE} Conference on
  Computational Complexity, {CCC} 2001}, pages 131--137, 2001.
\newblock See \url{https://arxiv.org/abs/quant-ph/0007016}.

\bibitem[BDRV18]{DBLP:conf/eurocrypt/BermanDRV18}
Itay Berman, Akshay Degwekar, Ron~D. Rothblum, and Prashant~Nalini Vasudevan.
\newblock Multi-collision resistant hash functions and their applications.
\newblock In {\em Proceedings of the 37th Annual International Conference on
  the Theory and Applications of Cryptographic Techniques, {EUROCRYPT} 2018,
  Part {II}}, pages 133--161, 2018.
\newblock See \url{https://eprint.iacr.org/2017/489}.

\bibitem[Bel12]{Belovs12}
Aleksandrs Belovs.
\newblock Learning-graph-based quantum algorithm for $k$-distinctness.
\newblock In {\em Proceedings of the 53rd Annual {IEEE} Symposium on
  Foundations of Computer Science, {FOCS} 2012}, pages 207--216, 2012.
\newblock See \url{https://arxiv.org/abs/1205.1534v2}.

\bibitem[BHT97]{BHT97}
Gilles Brassard, Peter H{\o}yer, and Alain Tapp.
\newblock Quantum algorithm for the collision problem.
\newblock {\em CoRR}, quant-ph/9705002, 1997.
\newblock See also {Quantum Cryptanalysis of Hash and Claw-Free Functions. In
  {\it Proceedings of the Third Latin American Theoretical Informatics
  Symposium}, LATIN 1998, pages 163-169, 1998}. See
  \url{https://arxiv.org/abs/quant-ph/9705002}.

\bibitem[BKP18]{DBLP:conf/stoc/BitanskyKP18}
Nir Bitansky, Yael~Tauman Kalai, and Omer Paneth.
\newblock Multi-collision resistance: a paradigm for keyless hash functions.
\newblock In {\em Proceedings of the 50th Annual {ACM} {SIGACT} Symposium on
  Theory of Computing, {STOC} 2018}, pages 671--684, 2018.
\newblock See \url{https://eprint.iacr.org/2017/488}.

\bibitem[CN08]{DBLP:conf/fse/ChangN08}
Donghoon Chang and Mridul Nandi.
\newblock Improved indifferentiability security analysis of {chopMD} hash
  function.
\newblock In {\em Proceedings of the 15th International Workshop on Fast
  Software Encryption, FSE 2008}, pages 429--443, 2008.

\bibitem[DDKS14]{DBLP:conf/asiacrypt/DinurDKS14}
Itai Dinur, Orr Dunkelman, Nathan Keller, and Adi Shamir.
\newblock Cryptanalysis of iterated {E}ven-{M}ansour schemes with two keys.
\newblock In {\em Proceedings of the 20th International Conference on the
  Theory and Application of Cryptology and Information Security, {ASIACRYPT}
  2014, Part {I}}, pages 439--457, 2014.
\newblock See \url{https://eprint.iacr.org/2013/674}.

\bibitem[Flo67]{F67-1}
Robert~W. Floyd.
\newblock Nondeterministic algorithms.
\newblock {\em Journal of the ACM}, 14(4):636--644, 1967.

\bibitem[Gro96]{Gr}
Lov.~K Grover.
\newblock A fast quantum mechanical algorithm for database search.
\newblock In {\em {Proceedings of the 28th Annual {ACM} {SIGACT} Symposium on
  Theory of Computing, {STOC} 1996}}, pages 212--219, 1996.
\newblock See \url{https://arxiv.org/abs/quant-ph/9605043}.

\bibitem[HIK{\etalchar{+}}10]{DBLP:conf/icisc/HiroseIKOPY10}
Shoichi Hirose, Kota Ideguchi, Hidenori Kuwakado, Toru Owada, Bart Preneel, and
  Hirotaka Yoshida.
\newblock A lightweight 256-bit hash function for hardware and low-end devices:
  {Lesamnta}-{LW}.
\newblock In {\em Proceedings of the 13th International Conference on
  Information Security and Cryptology, ICISC 2010}, pages 151--168, 2010.

\bibitem[HS05]{hush2005concentration}
Don Hush and Clint Scovel.
\newblock Concentration of the hypergeometric distribution.
\newblock {\em Statistics \& probability letters}, 75(2):127--132, 2005.

\bibitem[HSTX19]{PQC19}
Akinori Hosoyamada, Yu~Sasaki, Seiichiro Tani, and Keita Xagawa.
\newblock Improved quantum multicollision-finding algorithm.
\newblock In {\em Proceedings of the 10th International Conference on
  Post-Quantum Cryptography, PQCrypto 2019}, pages 350--367, 2019.
\newblock See \url{https://eprint.iacr.org/2018/1122} and
  \url{https://arxiv.org/abs/1811.08097}.

\bibitem[HSX17]{DBLP:conf/asiacrypt/HosoyamadaSX17}
Akinori Hosoyamada, Yu~Sasaki, and Keita Xagawa.
\newblock Quantum multicollision-finding algorithm.
\newblock In {\em Proceedings of the 23rd International Conference on the
  Theory and Application of Cryptology and Information Security, {ASIACRYPT}
  2017, Part {II}}, pages 179--210, 2017.
\newblock See \url{https://eprint.iacr.org/2017/864}.

\bibitem[Jef14]{Jeffery14}
Stacey Jeffery.
\newblock {\em Frameworks for Quantum Algorithms}.
\newblock PhD thesis, University of Waterloo, 2014.

\bibitem[JJV02]{DBLP:conf/fse/JaulmesJV02}
{\'{E}}liane Jaulmes, Antoine Joux, and Fr{\'{e}}d{\'{e}}ric Valette.
\newblock On the security of randomized {CBC-MAC} beyond the birthday paradox
  limit: {A} new construction.
\newblock In {\em Proceedings of the 9th International Workshop on Fast
  Software Encryption, FSE 2002}, pages 237--251, 2002.
\newblock See \url{https://eprint.iacr.org/2001/074}.

\bibitem[JL09]{JL09}
Antoine Joux and Stefan Lucks.
\newblock Improved generic algorithms for $3$-collisions.
\newblock In {\em Proceedings of the 15th International Conference on the
  Theory and Application of Cryptology and Information Security,
  ASIACRYPT~2009}, pages 347--363, 2009.
\newblock See \url{https://eprint.iacr.org/2009/305}.

\bibitem[JLM14]{JLM14}
Philipp Jovanovic, Atul Luykx, and Bart Mennink.
\newblock Beyond $2^{c/2}$ security in sponge-based authenticated encryption
  modes.
\newblock In {\em Proceedings of the 20th International Conference on the
  Theory and Application of Cryptology and Information Security, ASIACRYPT
  2014, Part I}, pages 85--104, 2014.
\newblock See \url{https://eprint.iacr.org/2014/373}.

\bibitem[KMRT09]{DBLP:conf/eurocrypt/KnudsenMRT09}
Lars~R. Knudsen, Florian Mendel, Christian Rechberger, and S{\o}ren~S. Thomsen.
\newblock Cryptanalysis of {MDC-2}.
\newblock In {\em Proceedings of the 28th Annual International Conference on
  the Theory and Applications of Cryptographic Techniques, {EUROCRYPT} 2009},
  pages 106--120, 2009.

\bibitem[KNY18]{DBLP:conf/eurocrypt/KomargodskiNY18}
Ilan Komargodski, Moni Naor, and Eylon Yogev.
\newblock Collision resistant hashing for paranoids: Dealing with multiple
  collisions.
\newblock In {\em Proceedings of the 37th Annual International Conference on
  the Theory and Applications of Cryptographic Techniques, {EUROCRYPT} 2018,
  Part {II}}, pages 162--194, 2018.
\newblock See \url{https://eprint.iacr.org/2017/486}.

\bibitem[Kut05]{Kut05}
Samuel Kutin.
\newblock Quantum lower bound for the collision problem with small range.
\newblock {\em Theory of Computing}, 1:29--36, 2005.
\newblock See \url{https://arxiv.org/abs/quant-ph/0304162}.

\bibitem[LZ19]{DBLP:journals/iacr/LiuZ18}
Qipeng Liu and Mark Zhandry.
\newblock On finding quantum multi-collisions.
\newblock In {\em Proceedings of the 38th Annual International Conference on
  the Theory and Applications of Cryptographic Techniques, {EUROCRYPT} 2019,
  Part {III}}, pages 189--218, 2019.
\newblock See \url{https://eprint.iacr.org/2018/1096}.

\bibitem[MT08]{MT08}
Florian Mendel and S{\o}ren~S. Thomsen.
\newblock An observation on {JH}-512.
\newblock Available online, 2008.
\newblock See \url{http://ehash.iaik.tugraz.at/uploads/d/da/Jh_preimage.pdf}.

\bibitem[MU17]{mitzenmacher2017probability}
Michael Mitzenmacher and Eli Upfal.
\newblock {\em Probability and computing: Randomization and probabilistic
  techniques in algorithms and data analysis}.
\newblock Cambridge university press, 2017.

\bibitem[NO14]{DBLP:conf/scn/NaitoO14}
Yusuke Naito and Kazuo Ohta.
\newblock Improved indifferentiable security analysis of {PHOTON}.
\newblock In {\em Proceedings of the 9th International Conference on Security
  and Cryptography for Networks, SCN 2014}, pages 340--357, 2014.

\bibitem[NS16]{DBLP:conf/asiacrypt/NikolicS16}
Ivica Nikoli\'{c} and Yu~Sasaki.
\newblock A new algorithm for the unbalanced meet-in-the-middle problem.
\newblock In {\em Proceedings of the 22nd International Conference on the
  Theory and Application of Cryptology and Information Security, {ASIACRYPT}
  2016, Part {I}}, pages 627--647, 2016.
\newblock See \url{https://eprint.iacr.org/2016/851}.

\bibitem[NSWY13]{DBLP:conf/iwsec/NaitoSWY13}
Yusuke Naito, Yu~Sasaki, Lei Wang, and Kan Yasuda.
\newblock Generic state-recovery and forgery attacks on {ChopMD-MAC} and on
  {NMAC/HMAC}.
\newblock In {\em Proceedings of the 8th International Workshop on Security,
  IWSEC 2013}, pages 83--98, 2013.

\bibitem[NWW13]{DBLP:conf/fse/NikolicWW13}
Ivica Nikoli\'{c}, Lei Wang, and Shuang Wu.
\newblock Cryptanalysis of round-reduced {LED}.
\newblock In {\em Proceedings of the 20th International Workshop on Fast
  Software Encryption, FSE 2013}, pages 112--129, 2013.
\newblock See \url{https://eprint.iacr.org/2015/429}.

\bibitem[RS96]{RS96a}
Ronald~L. Rivest and Adi Shamir.
\newblock {PayWord and MicroMint} -- two simple micropayment schemes.
\newblock In {\em Proceedings of the 1996 International Workshop on Security
  Protocols, SPW 1996}, pages 69--87, 1996.

\bibitem[STKT08]{STKT08}
Kazuhiro Suzuki, Dongvu Tonien, Kaoru Kurosawa, and Koji Toyota.
\newblock Birthday paradox for multi-collisions.
\newblock {\em {IEICE} Transactions}, 91-A(1):39--45, 2008.
\newblock The preliminary version is in {ICISC 2006}.

\bibitem[Tan09]{DBLP:journals/tcs/Tani09}
Seiichiro Tani.
\newblock Claw finding algorithms using quantum walk.
\newblock {\em Theor. Comput. Sci.}, 410(50):5285--5297, 2009.
\newblock See \url{https://arxiv.org/abs/0708.2584}.

\bibitem[Yue14]{Yuen14}
Henry Yuen.
\newblock A quantum lower bound for distinguishing random functions from random
  permutations.
\newblock {\em Quantum Information {\&} Computation}, 14(13-14):1089--1097,
  2014.
\newblock See \url{https://arxiv.org/abs/1310.2885}.

\bibitem[Zha15]{Zha15}
Mark Zhandry.
\newblock A note on the quantum collision and set equality problems.
\newblock {\em Quantum Info. Comput.}, 15(7-8):557--567, 2015.
\newblock See \url{https://arxiv.org/abs/1312.1027}.

\end{thebibliography}

\end{document}